\documentclass[11pt,letterpaper]{article}

\usepackage[margin=1in]{geometry}

\usepackage{mathtools}
\usepackage{theorem,latexsym,graphicx,amssymb}
\usepackage{amsmath,enumerate}
\usepackage{wrapfig}
\usepackage{subfigure}
\usepackage{float}
\usepackage{psfig}
\usepackage{epsfig}
\usepackage{xspace}
\usepackage{paralist}
\usepackage[compact]{titlesec}
\usepackage[boxed,lined]{algorithm2e}
\usepackage{enumerate}
\usepackage{cases}
\usepackage{caption}
\usepackage{url}
\usepackage{enumitem}

\usepackage{color}
\definecolor{Darkblue}{rgb}{0,0,0.4}
\definecolor{Brown}{cmyk}{0,0.81,1.,0.60}
\definecolor{Purple}{cmyk}{0.45,0.86,0,0}
\usepackage[breaklinks]{hyperref}
\hypersetup{colorlinks=true,
            citebordercolor={.6 .6 .6},linkbordercolor={.6 .6 .6},%
citecolor=black,urlcolor=black,linkcolor=black,pagecolor=black}
\newcommand{\lref}[2][]{\hyperref[#2]{#1~\ref*{#2}}}


\usepackage{multicol}

\newenvironment{proof}{\noindent {\bf Proof:  }}{\hfill\rule{2mm}{2mm}}

\numberwithin{figure}{section}
\numberwithin{equation}{section}

\theoremstyle{definition}

\newtheorem{assumption}{Assumption}[section]

\theoremstyle{theorem}
\newtheorem{corollary}{Corollary}[section]
\newtheorem{theorem}{Theorem}[section]
\newtheorem{lemma}{Lemma}[section]
\newtheorem{claim}{Claim}[section]

\providecommand{\Appendix}{}
\renewcommand{\Appendix}[2][?]{%
        \refstepcounter{section}%
        \vspace{\parskip}%
        {\flushright\large\bfseries\appendixname\ \thesection: #1}%
        \vspace{\baselineskip}%
}

\renewcommand{\appendix}{%
        \newpage
        \renewcommand{\section}{\secdef\Appendix\Appendix}%
        \renewcommand{\thesection}{\Alph{section}}%
        \setcounter{section}{0}%
}

\newcommand{\ve}{\varepsilon}

\newcommand{\R}{\mathbb{R}}
\newcommand{\Z}{\mathbb{Z}}

\newcommand{\OPT}{\ensuremath{\mathsf{OPT}}\xspace}

\newcounter{note}[section]
\renewcommand{\thenote}{\thesection.\arabic{note}}

\newcommand{\initOneLiners}{%
    \setlength{\itemsep}{0pt}
    \setlength{\parsep }{0pt}
    \setlength{\topsep }{0pt}
}

\newcommand{\ignore}[1]{}

\newcommand*\samethanks[1][\value{footnote}]{\footnotemark[#1]}

\newcommand{\shortv}[1]{}

\newcommand{\defeq}{\coloneqq}

\newcommand{\copt}{\ensuremath{C^{\text{\rm opt}}}\xspace}
\newcommand{\popt}{\ensuremath{P^{\text{\rm opt}}}\xspace}

\newcommand{\hubert}[1]{\refstepcounter{note}$\ll${\sf Hubert's
Comment~\thenote:} {\sf \textcolor{red}{#1}}$\gg$\marginpar{\tiny\bf HC~\thenote}}

\title{Online Convex Covering and Packing Problems}

\author{T-H. Hubert Chan\thanks{Department of Computer Science, the University of Hong Kong. {\texttt{\{hubert,zhiyi,nkang\}@cs.hku.hk}}} \and Zhiyi Huang\samethanks \and Ning Kang\samethanks}
\date{}

\begin{document}

\begin{titlepage}

\maketitle

\begin{abstract}
\thispagestyle{empty}
We study the online convex covering problem and online convex packing problem.
The (offline) convex covering problem is modeled by the following convex program:
\[
\min_{\vec{x} \in \R_{+}^n} f(\vec{x}) \text{\rm ~s.t.~} A \vec{x} \ge \vec{1} 
\]
where $f : \R_{+}^n \mapsto \R_{+}$ is a monotone and convex cost function, and $A$ is an $m \times n$ matrix with non-negative entries.
Each row of the constraint matrix $A$ corresponds to a covering constraint.
In the online problem, each row of $A$ comes online and the algorithm must maintain a feasible assignment $\vec{x}$ and may only increase $\vec{x}$ over time.
The (offline) convex packing problem is modeled by the following convex program:
\[
\max_{\vec{y} \in \R_{+}^m} \textstyle \sum_{j = 1}^m y_j - g(A^T \vec{y}) 
\]
where $g : \R_{+}^n \mapsto \R_{+}$ is a monotone and convex cost function.
It is the Fenchel dual program of convex covering when $g$ is the convex conjugate of $f$.
In the online problem, each variable $y_j$ arrives online and the algorithm must decide the value of $y_j$ on its arrival.

\smallskip

We propose simple online algorithms for both problems using the online primal dual technique, and obtain nearly optimal competitive ratios for both problems for the important special case of polynomial cost functions.
For instance, for any convex homogeneous polynomial cost functions with non-negative coefficients and degree $\tau$, we introduce an $O(\tau \log n)^\tau$-competitive online convex covering algorithm, and an $O(\tau)$-competitive online convex packing algorithm, matching the known $\Omega(\tau \log n)^\tau$ \cite{AzarCP14} and $\Omega(\tau)$ \cite{HuangK15} lower bounds respectively.

\smallskip

There is a large family of online resource allocation problems that can be modeled under this online convex covering and packing framework, including online covering and packing problems (with linear objectives) \cite{BuchbinderN09a},  online mixed covering and packing \cite{AzarBFP13}, and online combinatorial
auction~\cite{BartalGN03,BlumGMS11}.
Our framework allows us to study these problems using a unified approach.
\end{abstract}

\thispagestyle{empty}
\end{titlepage}

\section{Introduction}

Due to its wide practical applications,
online optimization has received much attention in the communities of computer science, operation research, and applied mathematics.
In particular, in an online resource allocation problem, there is a set of resources and a set of requests, each of which demands certain amount of each resource.
The requests arrive one by one online, and the online algorithm must decide whether to serve a request and/or how to serve the request immediately on its arrival without information of future requests.

Most previous work on online resource allocation problems can be divided into two classes:
\emph{online covering problems}, where the online algorithm must satisfy all (covering) requests with an objective of minimizing a linear cost function for using the resources;
and \emph{online packing problems}, where the online algorithm aims to maximize the total number of served requests (or, more generally, the total value generated from the served requests) subject to fixed resource capacities.\footnote{Note that the other two possible combinations, i.e., covering problems with fixed resource capacities and packing with linear resource costs, have trivial $1$-competitive online algorithms.}
However, the two extreme cases of either having linear resource costs or having fixed resource capacities (i.e., zero-infinity step cost function) are far from capturing the whole story in practical scenarios.
For many typical resources, such as computing cycles, memory, labor, oil, etc., additional resources can be obtained at increasing marginal costs, which cannot be modeled in the traditional online covering and packing framework.
This observation motivates the natural question of how to solve online resource allocation problems with general convex cost functions.

On the other hand, there have been a few spontaneous studies on problems with mixed covering and packing constraints, e.g., online facility location \cite{AlonAABN06}, online unrelated machine scheduling \cite{AspnesAFPW97, AwerbuchAGKKV95}, and the recent work by Azar et al.~\cite{AzarBFP13} for an attempt to develop a unified framework for such problems.
Problems with both covering and packing constraints are traditionally considered to be more difficult than problems with only one types of constraints. 
Indeed, most known algorithms for mixed covering and packing problems and their analysis are more complicated compared to their counterparts in covering or packing problems.


\subsection{Our Contributions}

The first contribution of this paper is a conceptual one.
We propose an \emph{online convex covering and packing} framework that allows us to study a large family of online resource allocation problems with general convex cost functions under a unified framework.
Interestingly, we can also remodel the mixed covering and packing problem by Azar et al.~\cite{AzarBFP13}, which was originally modeled using linear programs, as a pure covering problem, but with a convex cost function.
Then, we are able apply our unified framework to the problem after remodelling and obtain a better competitive ratio with a simpler algorithm and analysis.

Our main technical contributions are two simple online algorithms for online convex covering and online convex packing respectively.
Our algorithms achieve nearly optimal competitive ratios for both problems for the important special case of polynomial cost functions (that are convex and have non-negative coefficients), where the competitive ratio is parameterized by the maximum degree of the cost function.
For online convex covering, we also show that our algorithm is competitive with an additive error for general cost functions.

We focus on polynomial cost functions for two reasons.
First, they are the most common type of cost functions when we model practical problems using convex programs.
For example, if we want to model energy minimization in online scheduling, all known theoretical studies (e.g., \cite{YaoDS95}) consider polynomial power functions.
There are also a large family of problems, e.g., the mixed covering and packing problem \cite{AzarBFP13}, whose objective can be modeled as the $\ell_p$-norm of a set of linear or polynomial functions; hence,
the $p$-th power of the objective is a polynomial.
Second, there are known online optimization problems for which polynomial cost functions are the ``worst'' cost functions which still admit bounded competitive ratios.
For example, Devanur and Huang \cite{DevanurH14} (Thm. 3.2 and 5.1) showed that the problem of minimizing flow-time plus energy admits bounded competitive ratios if and only if there is a parameter $\tau$ such that the power function is ``at most as convex as'' degree-$\tau$ polynomials.
We next discuss our framework and results.

\paragraph{Online Convex Covering:}

In an online covering problem, there are
 $n$ resources and let $x_i$ denote the amount of resource $i$ that are used by the algorithm; 
there are $m$ (covering) requests that arrive online, where each request $j$ is specified by a vector $\vec{a}_j = (a_{j1}, \dots, a_{jn})$, and has the form of $\sum_{i = 1}^n a_{ji} x_i \ge 1$. 
The goal is to minimize the resource costs subject to satisfying all requests.
Unlike the original online covering problem, where the cost function is linear in $\vec{x}$, in the online convex covering framework, we allow a general convex cost function $f$ such that producing $x_i$ amount of each resource $i$ for all $i = 1, \dots, n$ incurs a cost of $f(\vec{x})$.
Clearly, the design of online algorithms and the corresponding competitive analysis will depend on properties of the cost function $f$.
For polynomial cost functions, we show the following result.

\begin{theorem}[Online Convex Covering]
\label{thm:main_cover_poly}
Suppose the cost function $f$ is a convex multivariate polynomial
with non-negative coefficients and maximum degree $\tau$.  
\begin{compactitem}
\item[(a)]  If $f$ is a homogeneous polynomial, then there is an $O(\tau \log n)^\tau$-competitive online algorithm.

\item[(b)] If $f$ is a sum of $N$ monomials,
then there is an $O(\tau \log N)^\tau$-competitive online algorithm.
\end{compactitem}
%
%
%
%
%
\end{theorem}

The above result is nearly optimal due to a lower bound by Azar et al.~\cite[Theorem 3]{AzarCP14} showing that if the objective is the $\ell_p$-norm of a set of linear cost functions, the competitive ratio of any online algorithm is at least $\Omega(p \log n)$ (for sufficiently large $n$).  To compare with our results, we set $\tau = p$ and
take the $p$-th power of the $\ell_p$-norm to convert it 
into a degree-$\tau$ homogeneous polynomial such that the above lower bound
becomes $\Omega(\tau \log n)^\tau$.
For general cost functions, we show the following.

\begin{theorem}
\label{thm:main_cover_general}
Suppose the cost function $f$ is at most as convex as degree-$\tau$ polynomials in the sense that $\langle x, \nabla f(x) \rangle \le \tau f(x)$.
Suppose in our online convex covering algorithm, the primal vector
$\vec{x}$ is initially set to $\vec{L}$,
and let $\vec{U}$ be an upper bound of $\vec{x}$ (e.g., $U_i = \tfrac{1}{\min_{j \in [m]} a_{ji}}$); denote $\mu = \max_{i \in [n]} \tfrac{U_i}{L_i}$. 
Then, our online convex covering algorithm is $O(\tau \ln \mu)^\tau$-competitive with an additive error of $f(\vec{L})$.
\end{theorem}

As a concrete application of our framework and our online convex covering algorithm, we obtain an improved competitive ratio for the online mixed packing and covering problem.


\begin{itemize}[leftmargin=.5cm, topsep=0cm]
%
%
%
\item {\bf $\ell_p$-norm of linear functions:~}
In fact, let us consider a more general problem.
Consider $l$ vectors $\vec{c}_k \in \R^n_+$ ($k \in [l]$),
and define $\vec{\lambda}: \R^n_+ \rightarrow \R^l_+$ such that for $k \in [l]$, $\lambda_k(\vec{x}) := \langle \vec{c}_k, \vec{x} \rangle$.  For some $p \geq 1$,
the cost function is $\| \vec{\lambda}(\vec{x}) \|_p$, which is the $\ell_p$-norm of the $l$-dimensional vector $\vec{\lambda} (\vec{x})$.
This cost function is considered in~Azar et al.~\cite{AzarCP14},
and 
the online mixed covering and packing problem~\cite{AzarBFP13}
is the special case when $p = \infty$.

\medskip

Typically, the problem is characterized by a sparsity parameter $d$, which in our case is the maximum number of non-zero coordinates in each of the vectors $\vec{c}_k$'s.\footnote{The notions of sparsity in different work are somewhat different; see Section \ref{sec:concurrent} for details.}
For finite $p \geq 1$, we consider
the polynomial $f(\vec{x}) := \| \vec{\lambda}(\vec{x}) \|_p^p$,
whose degree is at most $p$.  Observe that since each vector $\vec{c}_k$ has at most $d$ non-zero coordinates,
each term $\lambda_k(\vec{x})^p$ introduces at most $d$ variables.
Hence, $f(\vec{x})$ is a degree-$p$ homogeneous polynomial that depends on at most $ld$ variables.
However, as we shall see in Corollary~\ref{cor:lp_norm}, a careful application of Theorem~\ref{thm:main_cover_poly} can achieve an 
$O(p \log d)^p$-competitive online algorithm with respect to the cost function $f$.
Taking the $p$-th root, the competitive ratio for the $\ell_p$-norm is $O(p \log d)$,
simplifying the competitive analysis of Azar et al.~\cite{AzarBFP13, AzarCP14} by removing the logarithmic dependency on
the ratios of non-zero coordinates in $\vec{a}_j$'s and $\vec{c}_k$'s.


Since the dimension of vector $\vec{\lambda}(\vec{x})$ is $l$, it is well known that its $\ell_\infty$-norm can be approximated within constant factor by the $\ell_{\log l}$-norm.  Hence, using $p = \log l$ in our framework, we obtain a competitive ratio of $O(\log l  \log d)$ for the online mixed covering and packing problem.

\ignore{
In the online mixed covering and packing problem \cite{AzarBFP13}, other than the offline resources and online (covering) requests, there is also a set of $l$ packing constraints of the form $\sum_{i = 1}^n b_{ki} x_i \le \lambda_k$ for $k = 1, \dots, l$, where $\lambda_k$ are themselves variables denoting the ``congestion'' of packing constraint $k$.
The objective is to minimize the maximum congestion, i.e., $\max_{k \in [l]} \lambda_k$, or more generally, the $\ell_p$ norm of the congestions, i.e., $( \sum_{k = 1}^l \lambda_k^p )^{1/p}$.
Clearly, these objectives can be modeled as a convex function of $\vec{x}$ and, thus, fall into our online convex covering framework.

\smallskip

Using our framework, we get an $O(\log n \log l)$-competitive online algorithm minimizing the maximum congestion, and an $O(p \log n)$-competitive online algorithm for minimizing the $\ell_p$-norm, improving the competitive ratio of Azar et al.~\cite{AzarBFP13, AzarCP14} by removing the logarithmic dependency on the parameters $a_{ji}$'s and $b_{ki}$'s.
}
%
\end{itemize}

\paragraph{Online Convex Packing:}

In an online convex packing problem, there are $n$ resources and let $z_i$ denote the amount of resource $i$ that is used by the algorithm.
There are (packing) requests that arrive online, where each request is specified by a vector $\vec{a}_j = (a_{j1}, \dots, a_{jn})$ and $a_{ji}$ is the amount of resource $i$ needed for serving one unit of request $j$.
Unlike the original online packing problems, where the resources are subject to fixed capacities, we will allow a general convex cost function $g$ such that producing $z_i$ amount of each resource $i$ for all $i = 1, \dots, n$ incurs a resource cost of $g(\vec{z})$.
The goal is to maximize the total units of served requests minus the resource costs.
The design of online algorithms and the corresponding competitive analysis will depend on certain properties of the cost function $g$.
It turns out that online convex packing is the Fenchel dual problem of the online convex covering, when $g$ is the convex conjugate of $f$. See Section \ref{sec:prelim} for a brief discussion of convex conjugate and Fenchel duality.

\begin{theorem}
\label{thm:main_packing}
For any polynomial cost function $g$ that is convex, and has non-negative coefficients and maximum degree $\tau \ge 1$, our online convex packing algorithm is $O(\tau)$-competitive.
\end{theorem}


The above positive result also extends to a more complicated type of packing problems known as combinatorial auction with production cost.

\begin{itemize}[leftmargin=.5cm, topsep=0cm]
%
%
%
%
%
%
\item {\bf Online combinatorial auction with production cost:}
Online combinatorial auction is a natural extension of online packing. 
In an online combinatorial auction, each online request $j$ is a buyer associated with a value function $v_j$, where $v_j(\vec{y}_j)$ is buyer $j$'s value if he gets $y_{ji}$ units of resource $i$, $i \in [n]$, in the auction.
In other words, there are multiple ways of serving each request $j$; each way generates a different value and consumes a different amount of each resource.
The goal is to maximize the total value of the final allocation minus resource costs.
This problem was first studied in the fixed capacity setting \cite{BartalGN03, BuchbinderG13} and later in a setting with separable resource costs \cite{BlumGMS11, HuangK15}, i.e., the cost can be written as $g(\vec{z}) = \sum_{i = 1}^n g_i(z_i)$.

\smallskip

It is straightforward to extend our our results from online convex packing to online combinatorial auction with production cost.
We obtain an $O(\tau)$-competitive online algorithm for any convex polynomial cost functions with non-negative coefficients and maximum degree $\tau$.
This ratio is tight due to a lower bound result by Huang and Kim~\cite{HuangK15}.
\end{itemize}

\subsection{Our Techniques}

Our technique is based on the online primal dual framework for solving linear programs (see, e.g., \cite{BuchbinderN09b} for a survey), and its recent extensions to convex programs (e.g., \cite{DevanurH14, HuangK15, DevanurJ12, Huang14, GuptaKP13}).
The convex covering and packing problems are modeled by the following convex programs,
where $A$ is an $m \times n$ matrix with non-negative entries and $A^T$ is its transpose:
\begin{align*}
\text{\rm (Covering)} \quad \min_{\vec{x} \in \R_{+}^n} & ~~ f(\vec{x}) ~~ \text{subject to} ~~ A \vec{x} \ge \vec{1} \\
\text{\rm (Packing)} \quad \max_{\vec{y} \in \R_{+}^m} & \textstyle ~~ \sum_{j = 1}^m y_j - g(A^T \vec{y})
\end{align*}
As remarked in the previous subsection, these two problems are Fenchel dual programs of each other when $g$ is the convex conjugate of $f$, i.e., $g = f^*$ where $f^*(\vec{z}) \defeq \max_{\vec{x} \ge 0} \big\{ \langle \vec{x}, \vec{z} \rangle - f(\vec{z}) \big\}$.
In particular, weak duality holds: for any feasible covering assignment $\vec{x} \in \R_+^n$ such that $A \vec{x} \ge {1}$, and for any feasible packing assignment $\vec{y} \in \R_{+}^m$, we have $f(\vec{x}) \ge \sum_{j = 1}^m y_j - f^*(A^T \vec{y})$.

Our online primal dual algorithms simultaneously maintain feasible covering and packing assignments, namely, $\vec{x}$ and $\vec{y}$, such that (1) $\vec{x}$ and $\vec{y}$ are non-decreasing over time, and (2) at all times the covering objective $f(\vec{x})$ is at most $\alpha$ times the packing objective $\sum_{j = 1}^m y_j - g(A^T \vec{y})$ for some parameter $\alpha \ge 1$.
The first condition ensures that our algorithm is a feasible online algorithm for both the covering and packing problems.
The second condition, together with weak duality, shows that our online algorithm is $\alpha$-competitive for both the online convex covering problem with cost function $f$, and the online convex packing problem with cost function $f^*$.

While our framework gives online algorithms for both covering and packing simultaneously, we do design different algorithms for the two problems separately. 
This is because when the covering problem is of primary interest, we will exploit certain natural assumptions on the covering cost function $f$, such as monotone marginal cost (i.e., the gradient $\nabla f$ is monotone), to get good competitive ratios.
Similarly, when the packing problem if of primary interest, we will exploit assumptions such as $\nabla g = \nabla f^*$ being monotone to design and analyze our online algorithms.
Note that such natural assumptions on $f$ does not translate to natural assumptions on its conjugate $g = f^*$ in general.
In particular, $\nabla f$ being monotone does not imply that $\nabla f^*$ is monotone, and vice versa.\footnote{For instance, the function $f(x_1,x_2):= (x_1 + x_2)^2$ and its gradient $\nabla f$ are monotone on $\R^2_+$,
but its convex conjugate $f^*(z_1, z_2) = \max\{\frac{z_1^2}{4}, \frac{z_2^2}{4}\}$ is not differentiable on the line $z_1 = z_2$ and $\nabla f^*$ is not monotone. For example, when $z_2$ is increased from $(1,0)$ to $(1,2)$, $\nabla_1 f^*$ decreases from
$\nabla_1 f^*(1,0) = \frac{1}{2}$ to $\nabla_1 f^*(1,2) = 0$.}

\paragraph{Covering:}

Our online convex covering algorithm is very natural. 
On the arrival of a covering request $j$, our algorithm continuously increases $\vec{x}$ until request $j$ is satisfied, and for each resource~$i$, it increases $x_i$ exponentially at a rate proportional to the bang-per-buck ratio, i.e., $a_{ji}$ (fraction of request $j$ we can cover per unit of resource $i$) divided by $\nabla_i f(\vec{x})$ (the partial derivative of $f$ in the $i$-th coordinate).
In other words, the change of $x_i$ is proportional to $\tfrac{a_{ji}}{\nabla_i f(\vec{x})} x_i$.
Our algorithm generalizes the multiplicative update process for the original online covering problem with linear objectives \cite{BuchbinderN09a} and several other online optimization problems (e.g., \cite{BhawalkarGP14}) to handle general convex objectives.
See Section \ref{sec:covering} for details of our algorithm.

A novel technical ingredient of our approach is a new way to obtain competitive ratios that are independent of the entries of the constraint matrix $A$.
Consider the original online covering problem with linear objective as an example.
It is known that if we use the natural linear program relaxation and its dual, and we update both the covering and the packing assignments in a non-decreasing manner, then the competitive ratio is at least $\Omega(\log n + \log \tfrac{a_{\max}}{a_{\min}})$ (e.g., \cite{BuchbinderN09a}), where $a_{\max}$ and $a_{\min}$ are the maximum and minimum non-zero entries of the constraint matrix $A$.
Previous work such as Buchbinder and Naor~\cite{BuchbinderN09a} and Gupta and Nagarajan~\cite{GuptaN14} remove the dependency on the entries of $A$ by allowing the (dual) packing assignment to decrease.

We use a different approach.
We first observe that the extra dependency on the entries of $A$ is due to the linear part of the cost function---for any ``sufficiently convex'' cost function $f$, our online convex covering algorithm gives a competitive ratio that is independent of the entries of $A$.
Then, to obtain competitive ratios for any polynomial cost functions $f$, we first approximate the cost function $f$ with another surrogate polynomial cost function $\hat{f}$ that is ``sufficiently convex'', and then run our online convex covering algorithm with respect to $\hat{f}$.
See Section~\ref{sec:covering} for details.

We believe that the above technique, while simple, will find further applications in other online optimization problems.
In particular, it suggests an interesting and unexpected way of using our online convex covering and packing framework, as well as work of others, on extending online primal dual to convex programs. Even for problem with natural linear program relaxations, it may be helpful to remodel the problem using convex programs because convexity may improve the competitive ratios of online primal dual algorithms.

\paragraph{Packing:}

For the online convex packing problem, the (dual) covering assignment $\vec{x}$ can be interpreted as unit prices of the resources, where $x_i$ is the unit price of resource $i$.
In the offline optimal solution, the Karush-Kuhn-Tucker conditions (e.g., \cite{BoydV09}) indicate that the unit prices shall be equal to the marginal cost, i.e., $\vec{x} = \nabla g(A^T \vec{y})$.
In the online setting, however, the online algorithm does not know the final demand of resources $A^T \vec{y}$, and thus, does not know the final marginal cost $\nabla g(A^T \vec{y})$.
Instead, the online algorithm may predict the final demand given the current demand.
Our online algorithm uses a particularly simple prediction rule:
for some parameter $\rho > 1$, it predicts the final demand to be $\rho$ times the current demand and set $\vec{x} = \nabla g(\rho A^T \vec{y})$.
This is a natural generalization of the algorithms in \cite{BlumGMS11} and \cite{HuangK15}, which considers only separable cost functions that can be written as $g(\vec{z}) = g_1(z_1) + \dots + g_m(z_m)$.
It turns out this simple approach is able to obtain asymptotically tight competitive ratios for polynomial cost functions.
See Section \ref{sec:packing} for details.

\subsection{Related Work}

This paper employs the online primal dual technique which, informally speaking, utilizes the structure of the linear program relaxation of an optimization problem and its dual program to design and analyze online algorithms.
The technique has led to optimal or nearly optimal online algorithms for a large family of problems, e.g., online covering and packing (with linear objectives) \cite{BuchbinderN09a}, online caching \cite{BansalBN12}, online bipartite matching \cite{DevanurJK13}, etc.
Readers are referred to \cite{BuchbinderN09b} for a survey of results using the online primal dual technique.

In particular, our work adds into the recent online of work on extending the online primal dual technique to problems related to convex programs, e.g., online scheduling with speed-scaling \cite{DevanurH14, GuptaN14}, online matching with concave returns \cite{DevanurJ12}, online combinatorial auction with convex production cost \cite{HuangK15}, etc.
Before our work, the study of the online primal dual technique with convex programs has been on a problem by problem basis. 
Our online convex covering and packing framework can be viewed as a concrete step towards building a unified theory for using the online primal dual technique with convex programs.

Next, we discuss a few related work that are most related to our results.
Buchbinder and Naor \cite{BuchbinderN09a} studied the online covering and packing problem with linear objectives, which are special cases of our more general problems.
They proposed an $O(\log n)$-competitive online algorithm for online covering and an $O(\log n + \log \tfrac{a_{\max}}{a_{\min}})$-competitive online algorithm for online packing.

Azar et al.~\cite{AzarBFP13} proposed the online mixed packing and covering problem as an attempt to unify several online optimization problems with both covering and packing constraints, and proposed the first algorithm for this general problem with competitive ratio $O \big( \log d \log (l \beta \gamma) \big)$ where $\beta$ and $\gamma$ are parameters that depend on the covering and packing constraint matrices of the problem.
We obtain an improved competitive ratio of $O(\log d \log l)$ without the dependency on the parameters of constraint matrices.
Due to a lower bound result of Azar et al.~\cite{AzarCP14}, our competitive ratio is tight.

Blum et al.~\cite{BlumGMS11} proposed the study of online combinatorial auction with production cost, where the seller may produce multiple copies of each item subject to a convex cost function. 
In particular, for polynomial cost functions of degree $\tau$, Blum et al.~\cite{BlumGMS11} gave an $O(\tau)$-competitive online algorithm.
Later, Huang and Kim~\cite{HuangK15} improved the competitive ratios of Blum et al.~\cite{BlumGMS11} by constant factors and showed that the optimal competitive ratio can be characterized by a differential equation.
Note that the cost functions in both Huang and Kim~\cite{HuangK15} and Blum et al.~\cite{BlumGMS11} are defined on a per item basis, which corresponds to the special case of separable cost function that can be written as $g(\vec{z}) = g_1(z_1) + \dots + g_m(z_m)$ in our model.
In contrast, we give an $O(\tau)$-competitive online algorithm that can handle arbitrary degree-$\tau$ polynomial cost functions that may have correlation across different items.

\subsection{Concurrent Work}
\label{sec:concurrent}

In concurrent works, Azar et al.~\cite{AzarCP14} and Buchbinder et al.~\cite{BuchbinderCGNN14} have also independently studied the same problems.
Below we sketch the main differences between our work and those of the other two groups.

Azar et al.~\cite{AzarCP14} have considered the online convex covering problem, but not the packing problem explicitly. 
Their main result is essentially the same as our Theorem \ref{thm:main_cover_general} for general cost functions.
We stress that a major part of our effort in this paper is spent on improving the competitive ratio for general cost functions by focusing on the important special case of polynomial cost functions.
We propose a simple online algorithm with nearly optimal competitive ratio for this case (Theorem~\ref{thm:main_cover_poly}).

Buchbinder et al.~\cite{BuchbinderCGNN14} have studied both online convex covering and packing problems.
Their competitive ratio for online convex covering has a rather involved dependency on the cost function $f$; it does not seem possible to directly compare our results with theirs.
While both the algorithms of Buchbinder et al.~\cite{BuchbinderCGNN14} and ours are based on the online primal dual framework, there are many subtle differences.
For example, the main algorithm of Buchbinder et al.~\cite{BuchbinderCGNN14} may decrease dual variables while our algorithm always maintain dual variables in a non-decreasing manner.

It is instructive to compare the results for the online mixed packing and covering problem, an application of the online convex covering framework which all three groups consider.
Azar et al.~\cite{AzarCP14} obtained an $O(\log l \log d \beta \gamma)$-competitive algorithm, where $\beta$ and $\gamma$ are the maximum-to-minimum ratio of the non-zero entries of $\vec{a}_j$'s and $\vec{c}_k$'s respectively, and both $\vec{a}_j$'s (covering constraint vectors) and $\vec{c}_k$'s (packing constraint vectors) must be $d$-sparse in the sense that each vector may have at most $d$ non-zero entries.
Buchbinder et al.~\cite{BuchbinderCGNN14} obtained an $O(\log l \log d)$-competitive algorithm, removing the logarithmic dependency on $\beta$ and $\gamma$, and \emph{only the covering constraint vectors} need to be $d$-sparse.
In this paper, we propose an $O(\log l \log d)$-competitive algorithm, also removing the logarithmic dependency on $\beta$ and $\gamma$, and \emph{only the packing constraint vectors} need to be $d$-sparse.

For online convex packing, Buchbinder et al.~\cite{BuchbinderCGNN14} exploited the fact that it is the dual problem of online convex covering and directly use one of their online convex covering algorithms. 
As a result, their algorithms and competitive ratios are parameterized by the convex conjugate of the cost function. 
For some specific cost
functions, they managed to reformulate the problem and get around the
problem of having non-monotone gradient of the convex conjugate.
Readers are referred to their paper for details.

In contrast, we design simple online algorithms that are specialized for the online convex packing problem and the online combinatorial auction problem, and obtain asymptotically optimal competitive ratio for general convex polynomial cost functions (Theorem~\ref{thm:main_packing}).


\section{Preliminaries}
\label{sec:prelim}

For a positive integer $n$, we denote $[n] := \{1,2, \dots, n\}$.
We use $\vec{x}$ to denote a column vector.
For two vectors $\vec{a}$ and $\vec{b}$ of the same dimension, we write $\vec{a}\geq\vec{b}$ if each coordinate of $\vec{a}$ is at least
the corresponding coordinate of $\vec{b}$. We use
$\langle\vec{a},\vec{b}\rangle$ to denote the dot product of $\vec{a}$ and $\vec{b}$. 
Let $\vec{0}$ and $\vec{1}$ denote the all zero's and all one's vectors, respectively.  
A function $g:\R_{+}^{n} \rightarrow \R_{+}^m$ is monotone, if $\vec{a} \leq \vec{b}$ implies that $g(\vec{a}) \leq g(\vec{b})$.  
For $i \in [n]$, we use $\vec{e}_i$ to denote the unit vector whose $i$th coordinate is 1.

\subsection*{Conjugate}

Given a function $f:\R_{+}^{n} \rightarrow \R_{+}$, its conjugate $f^*:\R_{+}^{n} \rightarrow \R_{+}$ is defined as:
\[
f^{*}(\vec{z}) = \max_{\vec{x} \in \R_+^n} \big\{ \langle\vec{x},\vec{z}\rangle - f(\vec{x}) \big\}
\]
For example, if $f(x) = \tfrac{1}{\tau} x^\tau$ is a degree-$\tau$ polynomial ($\tau>1$), then $f^*(z) = (1 - \tfrac{1}{\tau}) z^{\tfrac{\tau}{\tau - 1}}$ is a degree-$\tfrac{\tau}{\tau - 1}$ polynomial.


Note that $\langle \vec{x}, \vec{z}\rangle - f(\vec{x})$ can be interpreted as the negative of the $y$-intercept of a hyperplane that passes through point $x$ and has gradient $z$.
Further, the negative of the $y$-intercept is maximized when the hyperplane is a tangent hyperplane of function $f$.
Hence, $f^*(z)$ can be interpreted as the $y$-intercept of the tangent hyperplane of $f$ that has gradient $z$.


For the rest of this paper, we will focus on functions $f$ that are non-negative, monotone and differentiable,  and $f(\vec{0}) = 0$. 
In this case, the conjugate satisfies the following properties:
%
\begin{compactitem}
\item The conjugate $f^*$ is non-negative, monotone, convex, and $f^*(\vec{0}) = 0$.
\item If $f$ is lower semi-continuous, $f^{**} = f$.
\end{compactitem}
Readers are referred to the textbook of Boyd and Vandenberghe~\cite[Chapter 3.3]{BoydV09} for detailed discussions on conjugate.

\subsection*{Online Convex Covering and Packing Problems}

\paragraph{Convex Covering:}

We will first introduce the offline convex covering problem.
Let there be $n$ resources.
Producing $x_i$ units of resource $i$ for each $i \in [n]$ incurs a production cost of $f(\vec{x})$, where $f:\R_{+}^{n} \rightarrow \R_{+}$ is convex, non-negative, monotone, and differentiable, and $f(\vec{0}) = 0$.
Let there be $m$ covering requests: $\sum_{i \in [n]} a_{ji} x_i \ge 1$ for $j \in [m]$, where $a_{ji}$'s are non-negative.
The objective is to minimize the total production cost $f(\vec{x})$ while satisfying all covering requests.

Let $\vec{a}_j$ denote the non-negative vector $(a_{j1}, a_{j2}, \dots, a_{jn})$ and $A = (a_{ji})$ denote the $m \times n$ matrix whose rows are $\vec{a}_j$'s.
The offline convex covering problem can be formulated as a convex program
\begin{equation}
\label{eq:covering}
\min_{\vec{x} \in \R_{+}^n} ~~ C(\vec{x}) \defeq f(\vec{x}) ~~ \text{subject to} ~~ A \vec{x} \ge \vec{1}
\end{equation}

In the online convex covering problem, the covering requests arrive one by one.
In round $k \in [m]$, covering request $k$ arrives
with the vector $\vec{a}_k \in \R^n_+$ (we use $j$ as a generic index of covering constraints and $k$ as the index of the current request), and the \emph{covering player} must decide immediately how to increase some of the $x_i$'s to satisfy the request $\sum_{i \in [n]} a_{ki} x_i \ge 1$ without knowing future requests. 
Note that the algorithm cannot decrease the values of $x_i$'s throughout the process.
We allow the covering requests to be chosen adversarially depending on the past response of the algorithm.
The goal is to approximately minimize the covering objective $C(\vec{x})$ while satisfying all covering requests. 

\paragraph{Convex Packing:}

We start with the offline problem.
Let there be $n$~resources.
Producing $z_i$~units of resource $i$ for each $i \in [n]$ incurs a production cost of $g(\vec{z})$, where $g:\R_{+}^{n} \rightarrow \R_{+}$ is convex, non-negative, monotone, and differentiable, and $g(\vec{0}) = 0$.
Let there be $m$ packing requests, where each request $j \in [m]$ is specified by a non-negative vector $\vec{a}_j = (a_{j1}, a_{j2}, \dots, a_{jn})$ such that serving one unit of packing request $j$ consumes $a_{ji}$ units of resource $i$.
The objective is to maximize the total number of units of packing requests served minus the total production cost.

Recall that $A = (a_{ji})$ denotes the $m \times n$ matrix whose rows are $\vec{a}_j$'s.
Let $y_j$ denote the number of units of packing request $j$ that the algorithm decides to serve.
Then, the total amount of resource~$i$ consumed is $z_i = \sum_{j \in [m]} a_{ji} y_j$.
The offline convex packing problem can be formulated as the following convex program.
\begin{equation}
\label{eq:packingorig}
\max_{\vec{y} \in \R_{+}^m} ~~ P(\vec{y}) \textstyle \defeq \sum_{j \in [m]} y_j - g \big( A^T \vec{y} \big)
\end{equation}

In the online problem, the packing requests arrive one by one.
In round $k \in [m]$, packing request $k$ arrives with the vector $\vec{a}_k$, and the \emph{packing player} must irrevocably pick a value $y_k \geq 0$.
We can also view that there is some vector $\vec{y}$ that is initially set to $\vec{0}$, and in round $k$, the online algorithm can only increase the $k$-th coordinate of $\vec{y}$.  
The goal is to approximately maximize the packing objective $P(\vec{y})$.

\paragraph{Convex Covering/Packing Duality:}

If the cost function of the convex packing problem is the conjugate of the cost function of the convex covering problem, namely, $g(\vec{z}) = f^*(\vec{z})$, the convex packing program \eqref{eq:packingorig} is the Fenchel dual program of the convex covering program \eqref{eq:covering}.
The reader is referred to~\cite{BoydV09} for details on Fenchel duality.

When we have explicit assumptions on $f$, we emphasize the
performance of the covering player, and let the packing player play
an auxiliary role.  On the other hand, when we have explicit assumptions on the conjugate $f^*$, the roles are reversed, 
and we put emphasize on the performance of the packing player.


\begin{lemma}[Weak Duality]
For any $\vec{x} \in \R_{+}^n$ such that $A \vec{x} \ge 1$ and any $\vec{y} \in \R_{+}^m$, we have that 
\[
C(\vec{x}) \ge P(\vec{y}), 
\]
which holds even if $f$ is not convex.
\end{lemma}

In the rest of this paper, we will use $f^*$ to denote the cost function in the convex packing problem, and the corresponding convex program becomes
\begin{equation}
\label{eq:packing}
\max_{\vec{y} \in \R_{+}^m} ~~ P(\vec{y}) \textstyle \defeq \sum_{j \in [m]} y_j - f^* \big( A^T \vec{y} \big)
\end{equation}
%






\paragraph{Competitive Ratio:}

We will compare the expectation of the objective value achieved by an online algorithm with the offline optimum.
Let $\copt$ and $\popt$ denote the offline optimum of the convex covering program and convex packing program respectively.
We use the convention that the competitive ratio is at least $1$.
Hence, for an online covering algorithm that returns $\vec{x}$, its ratio is $\frac{C(\vec{x})}{\copt}$; for an online packing algorithm that returns $\vec{y}$, its ratio is $\frac{\popt}{P(\vec{y})}$.


\subsection*{Online Primal Dual Framework}  

Readers are referred to, e.g., the survey by Buchbinder and Naor \cite{BuchbinderN09a}, for a comprehensive discussion of the online primal dual framework and its applications for problems with linear program relaxations, and, e.g., Devanur and Jain~\cite{DevanurJ12}, Devanur and Huang~\cite{DevanurH14}, etc., for its recent extensions to convex programs.
Below we briefly describe how the framework works for our problems.

\paragraph{Online Primal Dual Algorithm:}

Our online primal dual algorithms maintain both a feasible covering assignment $\vec{x}$ and a feasible packing assignment $\vec{y} \in \R_{+}^m$ online.
More precisely, we can view an online primal dual algorithm as running an online covering algorithm and an online packing algorithm simultaneously as follows:

\begin{itemize}[leftmargin=.5cm, topsep=0cm]
\item \textbf{Covering Algorithm.} 
A vector $\vec{x} \in \R^n_+$ is initially set to $\vec{0}$.  
In round $k \in [m]$, the covering player can increase some coordinates of $\vec{x}$ such that the covering constraint $\sum_{i \in [n]} a_{ki} x_i = \langle \vec{a}_k, \vec{x} \rangle \geq 1$ is satisfied.
The goal is to minimize $C(\vec{x}) := f(\vec{x})$ at the end of the process while satisfying all covering constraints.
\item \textbf{Packing Algorithm.} 
In round $k \in [m]$, the packing player irrevocably picks a value $y_k \geq 0$.
We can view that there is some vector $\vec{y}$ that is initially set to $\vec{0}$, and in round $k$, the packing player can only increase the $k$-th coordinate of $\vec{y}$.  
The goal is to maximize $P(\vec{y})$.
\end{itemize}

We note that there are online primal dual covering algorithms in literature that may decrease $y_j$'s.
For example, \cite{BuchbinderN09b} showed that decreasing $y_j$'s is crucial for obtaining competitive ratio that is independent of the parameters of the constraint matrix $A$ for the online covering with linear cost if we use the natural linear program relaxation.
In contrast, we obtain a competitive ratio independent of the parameters of the constraint matrix by considering an alternative convex program formulation the problem.
See Sections~\ref{sec:covering} and~\ref{sec:packing} for details.

\paragraph{Online Primal Dual Analysis:}

As the covering and the packing algorithms maintain their corresponding feasible solutions $\vec{x}$ and $\vec{y}$ such that after every round, we seek to preserve an invariant $C(\vec{x}) \leq \alpha \cdot P(\vec{y})$, where $\alpha$ is a parameter that can depend on the function $f$ (and its conjugate $f^*$) and $n$, but is independent on the number of requests $m$.

Because of weak duality $C(\vec{x}) \geq \copt \ge \popt \geq P(\vec{y})$,
it follows that  $C(\vec{x}) \leq \alpha \cdot P(\vec{y}) \leq \alpha \cdot \copt$,
and $\popt \leq C(\vec{x}) \leq \alpha \cdot P(\vec{y})$.
Hence, the online primal dual algorithm described above achieve a competitive ratio at most $\alpha$ for both the online covering and packing problems.

\section{Convex Online Covering}

\label{sec:covering}

Recall that the covering and the packing problems are dual of each other.
In this section, we will treat the convex covering problem as primal and the convex packing problem as dual.
%
%
\begin{align}
\text{\rm (Primal)} \quad \min_{\vec{x} \in \R_{+}^n} & ~~ C(\vec{x}) \defeq f(\vec{x}) ~~ \text{subject to} ~~ A \vec{x} \ge \vec{1}\\
\text{\rm (Dual)} \quad \max_{\vec{y} \in \R_{+}^m} & \textstyle ~~ P(\vec{y}) \defeq \sum_{j \in [m]} y_j - f^* \big( A^T \vec{y} \big)
\end{align}
%
%
Recall that for notational convenience, we denote $\vec{z} := A^T \vec{y}$.
We make the following assumptions on the cost function $f$ in the covering (primal) problem throughout this section.
The first set of assumptions are regularity assumptions that we have discussed in the preliminary. 

\begin{assumption}
\label{assumption:function_property}
The function $f$ is convex, differentiable, and monotone, and $f(\vec{0})=0$. 
\end{assumption}

Beyond the above regularity assumptions, we will further assume that the marginal cost is monotone.
Recall that the function $f$ models the production cost of resources.
This assumption is saying that producing more copies of some resource $i$ does not decrease the production cost of resource $i^\prime$ (regardless of whether $i^\prime = i$ or not).  In Section~\ref{sec:coveringpoly}, we will have an even stronger monotone condition $\nabla f$ to get a better competitive ratio.

\begin{assumption}
The gradient $\nabla f$ is monotone.
\end{assumption}

Finally, we introduce a parameter $\tau \ge 1$ that measures how convex the function $f$ is.

\begin{assumption}
\label{assumption:kappa}
There is a parameter $\tau \ge 1$ such that for all $\vec{x} \in \R^n_+$, $\langle \nabla f(\vec{x}), \vec{x} \rangle \leq \tau \cdot f(\vec{x})$.
\end{assumption}

For instance, if $f$ is a multivariate polynomial with non-negative coefficients, $\tau$ can be chosen as the maximum degree of $f$.
In particular, Assumption \ref{assumption:kappa} implies the following inequalities which will be useful in our competitive analysis.  We remark that the convexity of $f$ is only explicitly used
to prove part (c) of Lemma~\ref{lemma:convexf}.

\begin{lemma}[Convexity of $f$] \label{lemma:convexf}
Suppose that for some $\tau > 1$,
for all $\vec{x} \in \R^n_+$,
$\langle \nabla f(\vec{x}), \vec{x} \rangle \leq \tau \cdot f(\vec{x})$.
Then, the following statements hold.
\begin{compactitem}
\item[(a)] For $\delta \geq 1$, for $\vec{x} \in \R^n_+$,
$f(\delta \vec{x}) \leq \delta^\tau f(\vec{x})$.

\item[(b)] For $0 < \gamma \leq 1$, for $\vec{z} \in \R^n_+$,
$f^*(\gamma \vec{z}) \leq \gamma^{\frac{\tau}{\tau-1}} \cdot f^*(\vec{z})$.

\item[(c)] Suppose further that $f$ is convex.  Then, for all $\vec{x} \in \R^n_+$, 

$f^*(\nabla f(\vec{x})) = \langle \vec{x}, \nabla f(\vec{x}) \rangle - f(\vec{x})  \leq (\tau - 1) \cdot f(\vec{x})$. 
\end{compactitem}
In particular, statements (b) and (c) implies that
for $0 \leq \gamma \leq 1$, 
$f^*(\gamma \cdot \nabla f(\vec{x})) \leq \gamma^{\frac{\tau}{\tau-1}} \cdot  (\tau - 1) \cdot f(\vec{x})$.
\end{lemma}

\begin{proof}
For statement (a), define the function $g : [1, +\infty) \rightarrow \R$
by $g(\theta) := \ln f(\theta \vec{x})$.
Then, $g'(\theta) = \frac{\langle f(\theta \vec{x}), \vec{x}\rangle}{f(\theta \vec{x})} \leq \frac{\tau}{\theta}$,
where the last inequality follows because 
$\langle f(\theta \vec{x}), \theta \vec{x}\rangle \leq \tau f(\theta \vec{x})$.
Integrating over $\theta \in [1, \delta]$ gives $g(\delta) - g(1) \leq \tau \ln \delta$, which is equivalent to $f(\delta \vec{x}) \leq \delta^\tau f(\vec{x})$.

For statement (b), for $0 < \gamma \leq 1$, let $\delta := (\frac{1}{\gamma})^{\frac{1}{\tau-1}} \geq 1$, we have
\begin{eqnarray*}
f^{*}(\gamma \vec{z}) & = & \max\limits_{\vec{x}\in R_{+}^{n}} \{\langle\vec{x},\gamma \vec{z}\rangle - f(\vec{x}) \} = \gamma^{\frac{\tau}{\tau-1}}\max\limits_{\vec{x}\in R_{+}^{n}} \{\langle  \delta \vec{x}, \vec{z}\rangle - \delta^\tau f(\vec{x}) \} \\
& \le & \gamma^{\frac{\tau}{\tau-1}}\max\limits_{\vec{x}\in R_{+}^{n}} \{\langle  \delta \vec{x}, \vec{z}\rangle - f(\delta \vec{x}) \} = \gamma^{\frac{\tau}{\tau-1}} f^*(\vec{z}) \enspace,
\end{eqnarray*}
where the inequality follows from statement (a).

For statement (c), observe that
$f^* (\nabla f(\vec{x})) = \max_{\vec{w} \in \R^n_+} h(\vec{w})$,
where $h(\vec{w}) := \langle \vec{w}, \nabla f(\vec{x}) \rangle - f(\vec{w})$.
Hence, $\nabla h(\vec{w}) =  \nabla f(\vec{x}) - \nabla f(\vec{w})$.
Since $f$ is a convex function, $h$ is a concave function and $\nabla h(\vec{x}) = \vec{0}$;
it follows that $h$ is maximized when $\vec{w} = \vec{x}$.
Therefore, we have $f^* (\nabla f(\vec{x})) = \langle \vec{x}, \nabla f(\vec{x}) \rangle - f(\vec{x}) \leq (\tau - 1) \cdot f(\vec{x})$,
where the last inequality follows from the assumption on $f$.
\end{proof}

\subsection{Online Covering Algorithm} 

The details are described below in Algorithm~\ref{alg:covering}. 

\begin{algorithm}[H]
\label{alg:covering}
\SetAlgoLined
\textbf{Offline Input:} function $f : \R_+^n \rightarrow \R$

\textbf{Initialize:} $\vec{x} :=\vec{L}$ and $\vec{z}:=\vec{0}$ \;
\While{constraint vector $\vec{a}_k = (a_{k1},
\ldots, a_{kn})$ arrives in round $k$}{
	Set $y_{k} :=0$ \;
	\While{$\sum_{i=1}^{n}a_{ki}x_{i}<1$}{
		Continuously increase $y_k$: \\
		Simultaneously for each $i \in [n]$, increase $x_i$ at rate $\frac{d x_i}{d y_k}=\frac{\rho a_{ki}x_i}{\nabla_{i}f(\vec{x})}$, and \\
		increase $z_i$ at rate $\frac{d z_i}{d y_k} = a_{ki}$.\\
		
	}
}
\caption{Convex online covering}
\end{algorithm}

The algorithm takes the objective function $f$ of the covering problem as input,
but as we shall later see, sometimes the algorithm is run on a modified version of $f$ that might not even be convex.
The vector $\vec{x}$ in the covering problem is the primal variable, and is initialized to some appropriate $\vec{L} \in \R^n_+$.
In each round $k \in [m]$, a new constraint vector $\vec{a}_k = (a_{k1}, a_{k2}, \ldots, a_{kn})$ is presented to the algorithm.
The algorithm increases $x_i$'s continuously in the following manner until the new constraint $\sum_{i \in [n]} a_{ki} \ge 1$ is satisfied:

Imagine the dual variable $y_k$ as time. 
As $y_k$ increases, the algorithm increases variable $x_i$ with rate $\frac{d x_i}{d y_k} = \frac{\rho a_{ki} x_i}{\nabla_i f(\vec{x})}$ simultaneously for all $i \in [n]$, where $\nabla_i f(\vec{x}) = \frac{\partial f(\vec{x})}{\partial x_i}$ is the partial derivative of $f$ with respect to $x_i$, and $\rho > 0$ is some parameter to be determined.  
(We assume that the $x_i$'s are increased as stated, without worrying about how the differential equations are solved.)
In order to maintain the auxiliary vector $\vec{z}= A^T \vec{y}$, we initialize $\vec{z}$ to $\vec{0}$, and in round $k \in [m]$, for each $i \in [n]$, we also increase $z_i$ at rate $\frac{d z_i}{d y_k} = a_{ki}$ while $y_k$ increases.

To show feasibility of the algorithm, we need to prove that all covering constraints are eventually satisfied. 
(Note that dual feasibility is trivial given our rule of updating $\vec{y}$ and $\vec{z}$.)
We will initialize $\vec{x}$ appropriately to ensure that at least one of $x_i$'s with $a_{ki} > 0$ will increase as $y_k$ increases and eventually $\sum_{i \in [n]} a_{ki}$ reaches $1$, at which moment the new constraint is satisfied and $y_k$ stops increasing.
In Section~\ref{sec:mono_strict},
we shall show in Lemma~\ref{lem:mono_strict_general}
that
if $\vec{x}$ is initialized such
that each coordinate is non-zero, then
the covering constraint will eventually be satisfied.
Lemma~\ref{lemma:mono_strict} shall discuss the case
when $\vec{x}$ is initialized to $\vec{0}$ with
additional assumptions on $f$.


\subsection{Competitive Analysis}

We will show two positive results regarding our online convex covering algorithm.
The first result holds for general cost functions that are at most as convex as degree-$\tau$ polynomials in the sense that $\langle \nabla f(\vec{x}), \vec{x} \rangle \le \tau f(\vec{x})$ for all $\vec{x} \in \R_{+}^n$.
We prove that our algorithm is competitive with an additive loss of $\tau \cdot C(\vec{L})$, where $\vec{L}$ is the initial value of $\vec{x}$.
The precise statement is presented below.

\begin{theorem}
\label{thm:coveringgeneral}
Recall that $\vec{L}$ is the initial value of $\vec{x}$.
Suppose that $\vec{x} \le \vec{U}$ throughout the algorithm.\footnote{For instance, $x_i \le \frac{1}{\min_{k \in [m]} a_{ki}}$ for any $i \in [n]$. So $U_i = \frac{1}{\min_{k \in [m]} a_{ki}}$ is a feasible upper bound.}
Then, letting $\mu = \max_{i \in [n]} \frac{U_i}{L_i} \ge 1$, and $\rho = \tau^{\tau-1} \cdot (\ln \mu)^\tau$, after the last round $m$, Algorithm \ref{alg:covering} returns a primal vector $\vec{x}^{(m)}$ such that $C(\vec{x}^{(m)}) \leq (\tau \ln \mu)^\tau \cdot \copt + \tau \cdot C(\vec{L})$.
\end{theorem}

Our second positive result focuses on polynomial cost functions, an important special case in the literature as almost all known theoretical studies on related topics consider polynomial cost functions.
The precise statement is presented below.

\begin{theorem}[Polynomial with Max Degree $\tau$]
\label{thm:coveringpoly}
Suppose $f:\R^n_+ \rightarrow \R$ is a convex multivariate polynomial
with non-negative coefficients and maximum degree $\tau$.  Then, we have
the following.

\begin{compactitem}
\item[(a)]  Suppose the cost function $f$ is a homogeneous polynomial, i.e, only monomials
with degree exactly $\tau$ can have non-zero coefficients. 
Then, there is an $O(\tau \log n)^\tau$-competitive online covering algorithm.

\item[(b)] Suppose that $f$ can be expressed as a sum of $N$ monomials.
Then, there is an $O(\tau \log N)^\tau$-competitive online covering algorithm.
\end{compactitem}
\end{theorem}


Azar et al.~\cite{AzarCP14} proved a lower bound of $\Omega(\tau \log n)^\tau$ (for sufficiently large $n$) for the special case when $f(\vec{x})$ is the $\tau$-th power of the $\ell_\tau$-norm of a set of linear cost functions,
which is a homogeneous polynomial with degree $\tau$.
Hence, the competitive ratio of our online algorithm is nearly optimal.


\paragraph{Outline of Competitive Analysis.}  
After every round, we wish to give an upper bound of $C(\vec{x})$ in terms of $P(\vec{y}) = \sum_{j} y_j - f^*(\vec{z})$, where $\vec{z} = A^T \vec{y}$.  
Since $P(\vec{y})$ is the difference between the two terms $\sum_{j} y_j$ and $f^*(\vec{z})$, our analysis will handle them separately:

\begin{itemize}[leftmargin=.5cm, topsep=0cm]
\item In every round $k \in [m]$, we will compare the increase in $f(\vec{x})$ with the increase in $y_k$ (from 0) and show that the formal is upper bounded by the latter up to a factor of $\rho$ (Lemma \ref{lemma:yandf}).  
The analysis of this part is identical for both results.
It only requires that at each round, the algorithm can increase $\vec{x}$ according to the rules to satisfy the online covering constraint.
\item At the end of the process, we will give a coordinate-wise upper bound for $\vec{z}$.  
Since $f^*$ is monotone, this gives a lower bound for $P(\vec{y})$.  
We will exploit additional properties of $f$ that lead to a better upper bound to prove Theorem~\ref{thm:coveringpoly} for this part.
\end{itemize}

\subsubsection{Comparing Increases in $f(\vec{x})$ and $y_k$ in Each Round}

\begin{lemma}
\label{lemma:yandf}  
For $k \in [m]$, let $\vec{x}^{(k)}$ denote the primal
vector at the end of round $k$, where $\vec{x}^{(0)} = \vec{L}$ is
the initial vector.
Then, for each $k \in [m]$, at the end of round $k$, we have 
\[
\textstyle
y_k \geq \frac{1}{\rho} \cdot (f(\vec{x}^{(k)}) - f(\vec{x}^{(k-1)})) \enspace.
\]
In particular, summing this inequality over all rounds $k \in [m]$,
at the end of the process we have
\[
\textstyle
\sum_{k \in [m]} y_k \geq \frac{1}{\rho} \cdot (f(\vec{x}^{(m)}) - f(\vec{L})) \enspace.
\]
\end{lemma}

\begin{proof}
Note that in round $k \in [m]$, the algorithm increases $y_k$ only when $\sum_{i=1}^{n}a_{ki}x_{i}<1$. 
Therefore, while $y_k$ is increased, we have:
\begin{equation} \label{eq:yinc}
		1  \geq \sum_{i=1}^{n}a_{ki}x_{i} 
		       = \frac{1}{\rho} \sum_{i=1}^{n} \nabla_i f(\vec{x}) \cdot \frac{d x_i}{d y_k},
\end{equation}
where the equality comes from $\frac{d x_i}{d y_k} = \frac{\rho a_{ki} x_i}{\nabla_i f(\vec{x})}$.

As $y_k$ increases from 0 until the end of round $k$,
$\vec{x}$ increases from $\vec{x}^{(k-1)}$ to $\vec{x}^{(k)}$.
Hence, integrating both sides of (\ref{eq:yinc}) with respect to $y_k$ over the change during round $k$,
we have:
\[
y_k \geq \frac{1}{\rho} \int_{\vec{x} = \vec{x}^{(k-1)}}^{\vec{x}^{(k)}} \sum_{i \in [n]} \nabla_i f(\vec{x}) \, d x_i = \frac{1}{\rho} \int_{\vec{x} = \vec{x}^{(k-1)}}^{\vec{x}^{(k)}} \langle \nabla f(\vec{x}), d \vec{x} \rangle = \frac{1}{\rho} \cdot (f(\vec{x}^{(k)}) - f(\vec{x}^{(k-1)})) \enspace,
\]
where the last equality comes from the fundamental theorem of calculus for path integrals of vector fields.
	%
	%
\end{proof}


\subsubsection{Proof of Theorem \ref{thm:coveringgeneral}}


We will first show an upper bound for $\vec{z} = A^T \vec{y}$ in terms of $\nabla f$.  

\begin{lemma}
\label{lemma:boundz1}
After the last round $m$, we have 
$\vec{z}^{(m)} \leq \frac{\ln \mu}{\rho} \cdot \nabla f(\vec{x}^{(m)})$.
\end{lemma}

\begin{proof}
In each round $k \in [m]$, as $y_k$ increases, both $x_i$ and $z_i$ increases.
Specifically, 
\[
\frac{d z_i}{d x_i} = \frac{d z_i}{d y_k} \cdot \frac{d y_k}{d x_i} = a_{ki} \cdot \frac{\nabla_i f(\vec{x})}{\rho a_{ki} x_i} = \frac{\nabla_i f(\vec{x})}{\rho x_i} \leq
\frac{\nabla_i f(\vec{x}^{(m)})}{\rho x_i} \enspace,
\]
where the last inequality follows from the monotonicity of $\nabla_i f(\cdot)$.  Observe that as $\vec{z}$ increases from $\vec{0}$
to $\vec{z}^{(m)} = (z^{(m)}_{1}, z^{(m)}_{2}, \ldots, z^{(m)}_{n})$ at the end of round $m$, each $x_i$ increases from $L_i$ to $x^{(m)}_{i} \leq U_i$.
Hence, integrating $\frac{d z_i}{d x_i} \leq 
\frac{\nabla_i f(\vec{x}^{(m)})}{\rho x_i}$ with respect to $x_i$ throughout the whole process, we have
\[
z^{(m)}_{i} \leq \frac{\ln \frac{x^{(m)}_{i}}{L_i}}{\rho} \cdot \nabla_i f(\vec{x}^{(m)}) \leq \frac{\ln \frac{U_i}{L_i}}{\rho} \cdot \nabla_i f(\vec{x}^{(m)}) \leq \frac{\ln \mu}{\rho} \cdot \nabla_i f(\vec{x}^{(m)}) \enspace,
\]
as required.
\end{proof}

\begin{proof}[Theorem \ref{thm:coveringgeneral}]
Our choice of $\rho = \tau^{\tau-1} \cdot (\ln \mu)^\tau$ satisfies $\tfrac{\ln \mu}{\rho} \le 1$.
Hence, by the above lemma, the monotonicity of $f^*$, and the conclusion of Lemma~\ref{lemma:convexf},
we have 
\[
f^*(\vec{z}^{(m)}) \leq f^* \big( \tfrac{\ln \mu}{\rho} \cdot \nabla f(\vec{x}) \big) \le \big( \tfrac{\ln \mu}{\rho} \big)^{\frac{\tau}{\tau - 1}} \cdot (\tau - 1) 
\cdot C(\vec{x}^{(m)}) \enspace.
\]

Further note that from Lemma~\ref{lemma:yandf},
at the end of round $m$, the dual vector $\vec{y}$
satisfies 
\[
\sum_{k \in [m]} y_k \geq \tfrac{1}{\rho} \cdot (C(\vec{x}^{(m)}) - C(\vec{L})) \enspace.
\]

Hence, by $\rho = \tau^{\tau-1} \cdot (\ln \mu)^\tau$, we have
\[
P(\vec{y}) = \sum_{k \in [m]} y_k - f^*(\vec{z}^{(m)}) \geq \big( \tfrac{1}{\rho} - (\tfrac{\ln \mu}{\rho})^{\frac{\tau}{\tau - 1}} \cdot (\tau - 1) \big) \cdot C(\vec{x}^{(m)})  - \tfrac{C(\vec{L})}{\rho} 
= \tfrac{C(\vec{x}^{(m)})}{(\tau \ln \mu)^\tau} - \tfrac{\tau \cdot C(\vec{L})}{(\tau \ln \mu)^\tau} \enspace,
\]
which implies the result after rearranging, since $P(\vec{y}) \leq \popt \leq \copt$.
\end{proof}

\subsubsection{Proof of Theorem \ref{thm:coveringpoly}}
\label{sec:coveringpoly}

In this subsection, we consider the case when $f$ is a multivariate polynomial with non-negative coefficients.  
If $f$ has maximum degree $\tau$, then for all $\vec{x} \in \R^n_+$, $\langle \nabla f(\vec{x}), \vec{x} \rangle \leq \tau \cdot f(\vec{x})$.

We will first show how to obtain better competitive ratio if the gradient $\nabla f$ satisfies a stronger form of monotonicity.
Next, we will explain how to approximate any polynomial cost function $f$ with a surrogate cost function $\hat{f}$ such that its gradient $\nabla \hat{f}$ satisfies the stronger form of monotonicity.
Finally, we obtain the competitive ratio in Theorem \ref{thm:coveringpoly} by running our online covering algorithm with respect to the surrogate cost function $\hat{f}$.

\paragraph{Better Competitive Ratio from Stronger Monotonicity of $\nabla f$.}
For $\lambda \geq 0$, we say that $\nabla f$ is \emph{$\lambda$-monotone}, if
for each $i \in [n]$, the function $\vec{x} \mapsto \frac{\nabla_i f(\vec{x})}{x_i^\lambda}$ is monotone.
For example, if each $x_i$ appears with degree at least $(\lambda + 1)$ in $f$, then $\nabla f$ is $\lambda$-monotone.  With this assumption,
we prove a different version of Lemma~\ref{lemma:boundz1}

\begin{lemma}
\label{lemma:boundz2}
Suppose that $\nabla f$ is $\lambda$-monotone for some $\lambda > 0$.
Then, after the last round $m$ of Algorithm 1 running with $f$,
we have 
$\vec{z}^{(m)} \leq \frac{1}{\lambda \rho} \cdot \nabla f(\vec{x}^{(m)})$.  This is true even if $f$ is not convex.
\end{lemma}

\begin{proof}
Similar to the proof of Lemma \ref{lemma:boundz1}, we give an upper bound on the vector $\vec{z} = A^T y$ after
the end of round $m$, which we denote by $\vec{z}^{(m)}$.
Recall that we denote the primal vector after round $m$ by $\vec{x}^{(m)}$.
For each $i \in [n]$, 
\[
\textstyle
\frac{d z_i}{d x_i} = \frac{d z_i}{d y_k} \cdot \frac{d y_k}{d x_i} = a_{ki} \cdot \frac{\nabla_i f(\vec{x})}{\rho a_{ki} x_i} = \frac{1}{\rho \lambda} \cdot \frac{\nabla_i f(\vec{x})}{x_i^\lambda}
\frac{d x_i^\lambda}{d x_i} \leq
\frac{1}{\rho \lambda} \cdot \frac{\nabla_i f(\vec{x}^{(m)})}{x_{mi}^\lambda}
\frac{d x_i^\lambda}{d x_i} \enspace,
\]
where the last inequality follows from the $\lambda$-monotonicity of $\nabla f$ and $\vec{x} \leq \vec{x}^{(m)}$.

Observe that as $z_i$ increases from 0 to $z^{(m)}_{i}$,
$x_i$ increases from 0 to $x^{(m)}_{i}$.  Hence, integrating
with respect to $x_i$ throughout the algorithm,
we have for each $i$, $z^{(m)}_{i} \leq \frac{1}{\lambda \rho} 
\cdot \nabla_i f(\vec{x}^{(m)})$.
\end{proof}

\begin{lemma}
\label{lemma:sharp_convex}
Suppose that $f$ is convex, and for all $\vec{x} \in \R^n_+$, $\langle \nabla f(\vec{x}), \vec{x} \rangle \leq \tau \cdot f(\vec{x})$.
Moreover, $\nabla f$ is $\lambda$-monotone for some $0 < \lambda \leq \tau$.
Then, Algorithm~\ref{alg:covering} with $\rho := \frac{\tau^{\tau-1}}{\lambda^\tau}$ and initial primal vector $\vec{L} := \vec{0}$ is $(\frac{\tau}{\lambda})^\tau$-competitive.
\end{lemma}

\begin{proof}
We choose $\rho := \frac{\tau^{\tau-1}}{\lambda^\tau}$ such that $\gamma := \frac{1}{\lambda \rho} = (\frac{\lambda}{\tau})^{\tau - 1} \leq 1$.  
By Lemma~\ref{lemma:boundz2}, we have $\vec{z}^{(m)} \leq \gamma \cdot \nabla f(\vec{x}^{(m)})$.  Since $f$ is convex, we can apply Lemma~\ref{lemma:convexf}.
Observing that
$f^*$ is monotone,
we have 
\[
f^*(\vec{z}^{(m)}) \leq (\tau - 1) \cdot \gamma^{\frac{\tau}{\tau-1}} \cdot f(\vec{x}^{(m)}) = (\tau - 1) \cdot \big( \tfrac{\lambda}{\tau} \big)^{\tau} \cdot C(\vec{x}^{(m)}) \enspace.
\]

On the other hand, Lemma~\ref{lemma:yandf} states that
\[
\textstyle
\sum_{k \in [m]} y_k \geq \frac{1}{\rho} \cdot (f(\vec{x}^{(m)}) - f(\vec{L})) = \frac{\lambda^\tau}{\tau^{\tau-1}} \cdot C(\vec{x}^{(m)}) \enspace,
\]
since $f(\vec{L}) = f(\vec{0}) = 0$.

Therefore, putting things together,
we have 
\[
\textstyle
P(\vec{y}) = \sum_{k \in [m]} y_m - f^*(\vec{z}^{(m)})
\geq \frac{\lambda^\tau}{\tau^{\tau-1}} \cdot C(\vec{x}^{(m)}) - 
(\tau - 1) \cdot (\frac{\lambda}{\tau})^{\tau} \cdot C(\vec{x}^{(m)}) = (\frac{\lambda}{\tau})^\tau \cdot C(\vec{x}^{(m)}) \enspace.
\]
Recall that $P(\vec{y}) \leq \popt \leq \copt$, we have $C(\vec{x}^{(m)}) \leq (\frac{\tau}{\lambda})^\tau \cdot \copt$, as required.
\end{proof}

\paragraph{Homogeneous Polynomial $f$ with Degree $\tau$.} 

Next, we consider a degree-$\tau$ homogeneous polynomial $f$ with non-negative coefficients.
Note that any such polynomial satisfies our assumptions on the convexity of $f$, namely, for all $\vec{x} \in \R^n_+$,
$\langle \nabla f(\vec{x}) , \vec{x} \rangle \leq \tau \cdot f(\vec{x})$.
However, $\nabla f$ could be only $0$-monotone in general, in which case Lemma~\ref{lemma:sharp_convex} cannot be applied.

The idea is to approximate $f$ with another function $\widehat{f}$ such that $\nabla \widehat{f}$ is $\lambda$-monotone for some $\lambda > 0$. 
We will start with some notations for describing polynomials.
We use the vector $\vec{d} = (d_1, d_2, \ldots, d_n)  \in \Z^n$ to describe the degrees of the $x_i$'s in a monomial.  
For a monomial with maximum degree $\tau$, we have $\sum_{i \in [n]} d_i \leq \tau$ (where equality holds if the polynomial is homogeneous).  Observe that there are ${n + \tau \choose \tau} \leq n^\tau$ such vectors $\vec{d}$, for $n \geq 4$.
We denote $\vec{x}^{\vec{d}} := \prod_{i \in [n]} x_i^{d_i}$.
Given a vector $\vec{\alpha}$ with non-zero coordiniates,
we denote the vector $\frac{\vec{x}}{\vec{\alpha}} := (\frac{x_1}{\alpha_1}, \frac{x_2}{\alpha_2}, \ldots, \frac{x_n}{\alpha_n})$.

\begin{lemma}\label{lemma:lead_term}
For any homogeneous convex polynomial function $f$ with non-negative coefficients and degree $\tau$,
if $f$ has a monomial with positive coefficient that involves $x_i$,
then the term $x_{i}^{\tau}$ must also have positive coefficient.
\end{lemma}

\begin{proof}
Given $f: \R_+^n \rightarrow \R$, we fix some $i \in [n]$
and define $g: \R_+^2 \rightarrow \R$
as $g(x,y) := f( x \vec{e_{i}} + y(\vec{1}-\vec{e_{i}}))$,
i.e., we apply $f$ to the vector whose $i$-th coordinate is $x$
and all other coordinates are $y$.  Observe that since $f$ is convex,
so is $g$.

Suppose for contradiction's sake, $f$ depends on $x_i$, but
the coefficient of $x_i^\tau$ is zero.  Then, since $f$ is homogeneous with
degree $\tau$, it follows that $g$ can be expressed as
$g(x,y) = \sum_{j=0}^d b_j x^j y^{\tau - j}$,
where $1 \leq d < \tau$ and $b_d > 0$.

We next consider the Hessian matrix $H(g) = \left[
		\begin{matrix}
			\frac{\partial^2 g}{\partial x^2} & \frac{\partial^2 g}{\partial x \partial y} \\
			\frac{\partial^2 g}{\partial y \partial x} & \frac{\partial^2 g}{\partial y^2}
		\end{matrix}
		\right]$.
		
Because $g$ is convex, the matrix $H(g)$ is positive semi-definite,
and hence the determinant $|H(g)|$ is non-negative.
Observe that $|H(g)|$ is a polynomial in $x$ and $y$.
We shall reach a contradiction by showing that $|H(g)|$ is negative
for some $x$ and $y$.

For the case $d=1$, observe that $\frac{\partial^2 g}{\partial x^2} = 0$,
and $\frac{\partial^2 g}{\partial x \partial y} = \frac{\partial^2 g}{\partial y \partial x} > 0$ when $x$ and $y$ are both positive.  Hence,
the determinant is negative in this case.

For the case $2 \leq d < \tau$, we fix some $y > 0$ and
consider large $x > 0$.  Since $d < \tau$, the monomial in $|H(g)|$
with largest degree in $x$ is $x^{2(d-1)} y^{2(\tau - d -1)}$,
whose coefficient is
\[
b_d \cdot \{d(d-1) \cdot (\tau - d) (\tau - d -1) - d^2 (\tau - d)^2 \} < 0
\enspace.
\]

Therefore, for sufficiently large $x$, the determinant $|H(g)|$ is dominated
by the monomial $x^{2(d-1)} y^{2(\tau - d -1)}$, and hence is negative, reaching
the desired contradiction.
\end{proof}

\noindent \textbf{Approximating $f$.}
Suppose $f(\vec{x}) = \sum_{\vec{d}} \, c_{\vec{d}} \, \cdot \vec{x}^{\vec{d}}$.
In view of Lemma~\ref{lemma:lead_term},
we can assume, without loss of generality,
that $f$ depends on every $x_i$.  For
each $i \in [n]$, let $c_i > 0$ be the coefficient
of $x_i^\tau$ in $f$. Let the
vector $\vec{\alpha} \in \R_+^n$ be such that $\alpha_i := c_i^{\frac{1}{\tau}}$.

For $\lambda > 0$, define $\widehat{f} : \R_+^n \rightarrow \R$ by
$\widehat{f}(\vec{x}) := \sum_{\vec{d}} \, c_{\vec{d}} \cdot \vec{\alpha}^{\lambda \vec{d}} \, \cdot \vec{x}^{(1+\lambda)\vec{d}}$.

\begin{claim}
The function $\widehat{f}$ satisfies the following.

\begin{compactitem}
\item[1.] $\nabla \widehat{f}$ is $\lambda$-monotone;

\item[2.] $\langle \nabla \widehat{f} (\vec{x}), \vec{x} \rangle \leq \tau (1 + \lambda) \cdot \widehat{f}(x)$;

\item[3.] $\widehat{f}$ is convex.
\end{compactitem}
\end{claim}

\begin{proof}
The first statement follows because whenever $x_i$ appears, its degree is at least $1 + \lambda$.  The second statement
follows because each monomial has degree at most $(1+\lambda)\tau$.

We next prove that $\widehat{f}$ is convex.
Observe that $\widehat{f}(\vec{x}) = f(\varphi(x))$,
where $\varphi: \R_+^n \rightarrow \R_+^n$
is defined by $\varphi(\vec{x})_i := \alpha_i^\lambda \cdot x_i^{1+\lambda}$.

Since the mapping $t \mapsto t^{1+\lambda}$ is convex for $\lambda > 0$, it follows that for $p, q \geq 0$ such that $p+q=1$,
$\varphi(p \vec{x} + q \vec{y}) \leq p \cdot \varphi(\vec{x}) + q \cdot \varphi(\vec{y})$, where the inequality performs coordinate-wise comparisons.

Since $f$ is monotone,
we have $\widehat{f}(p  \vec{x} + q  \vec{y}) \leq f(p \cdot \varphi(\vec{x}) + q \cdot \varphi(\vec{y})) \leq p \cdot \widehat{f}(\vec{x}) + q \cdot \widehat{f}(\vec{y})$, where the last inequality follows
from the convexity of $f$.
\end{proof}

\ignore{

Therefore, it is always possible to transform $f(\vec{x})$ to one that each $x_{i}^{\tau}$ term has coefficient $1$ by setting $x_i$ to $(c_{\vec{e_i}})^{\frac{1}{\tau}}x_i$.

\begin{lemma}\label{lemma:coeffient_bound}
	After transformation, $\sum_{\vec{d}}c_{\vec{d}} \leq n^{\tau}$.
\end{lemma}
\begin{proof}
	Let $\vec{e_i}$ be the with $i$th dimension $1$ while all other dimensions $0$. From the convexity of $f$, we have
	\begin{eqnarray*}
		\sum_i f(n\vec{e_i}) \ge n \cdot f(\vec{1})
	\end{eqnarray*}
	Since the coefficient of each $x_i^\tau$ term is $1$, $f(n\vec{e_i})=n^\tau$ for any $i$, and $f(\vec{1})=\sum_{\vec{d}}c_{\vec{d}}$, the lemma is proved.
\end{proof}
}

\begin{lemma}[Approximation by $\widehat{f}$.]
\label{lemma:homo_approx}
For $0 < \lambda \leq 1$, suppose
$\widehat{f}$ is defined as above.
Then, we have $\frac{1}{2 n^{2\tau \lambda}} \cdot \widehat{f} \leq f^{1 + \lambda} \leq n^{\lambda \tau} \cdot \widehat{f}$.
%
%
\end{lemma}

\begin{proof}
Define $g : \R_+^n \rightarrow \R$
by $g(\vec{x}) := f(\frac{\vec{x}}{\vec{\alpha}})
= \sum_{\vec{d}} \,  b_{\vec{d}} \cdot \vec{x}^{\vec{d}}$,
where $b_{\vec{d}} = c_{\vec{d}} \cdot (\frac{\vec{1}}{\vec{\alpha}})^{\vec{d}}$.
Since $g$ is obtained from $f$ by scaling each $x_i$,
the convexity of $g$ follows from that of $f$.
In particular, observe that for each $i \in [n]$,
the coefficient of $x_i^\tau$ in $g$
is $b_i = c_i \cdot (\frac{1}{\alpha_i})^\tau = 1$.
Similarly, define $\widehat{g} : \R_+^n \rightarrow \R$
by $\widehat{g}(\vec{x}) := \widehat{f}(\frac{\vec{x}}{\vec{\alpha}})
= \sum_{\vec{d}} \,  b_{\vec{d}} \cdot \vec{x}^{(1+\lambda)\vec{d}}$.
Therefore, it suffices to prove the analogous inequality for $g$ and $\widehat{g}$.

We define $Z := \sum_{\vec{d}} b_{\vec{d}}$.
Observe that $Z = g(\vec{1}) \leq \frac{1}{n} \sum_{i \in [n]} g(n \vec{e}_i) = n^\tau$, where the inequality follows from the convexity of $g$, and the equality follows because $g$ is homogeneous and the coefficient of each $x_i^\tau$ is 1.


Using H\"{o}lder's inequality ($\sum_i s_i t_i \leq (\sum_i s_i^p)^{\frac{1}{p}} \cdot (\sum_i t_i^q)^{\frac{1}{q}}$
with $\frac{1}{p} = \frac{\lambda}{1+\lambda}$ and $\frac{1}{q} = \frac{1}{1 + \lambda}$), we have $g \le (\sum_{\vec{d}}b_{\vec{d}})^{\frac{\lambda}{1 + \lambda}} \cdot \widehat{g}^{\frac{1}{1+\lambda}}$. Thus,
we have the upperbound $g^{\lambda+1} \leq Z^\lambda \cdot \widehat{g} \le n^{\lambda\tau} \cdot \widehat{g}$.

Define $\theta := \frac{1}{n^{2\tau}}$,
and let $A := \{\vec{d} : b_{\vec{d}} \geq \theta Z\}$
index the coefficients that are at least the threshold $\theta Z \leq 1$.  Therefore, all the terms $x_i^\tau$
have coefficients 1 and pass the threshold.

Let $\overline{A}$ be the complement of $A$.
Observe that there are at most $n^\tau$ terms,
and for each $\vec{d}$, $\vec{x}^{(1+\lambda)\vec{d}} \leq \sum_{i\in[n]} x_i^{(1+\lambda)\tau}$.
Hence, 
\[
\sum_{\vec{d} \in \overline{A}} b_{\vec{d}} \cdot \vec{x}^{(1+\lambda)\vec{d}} 
\leq n^\tau \cdot \theta Z \cdot \sum_{i \in [n]} x_i^{(1+\lambda)\tau} \leq \sum_{\vec{d} \in A} b_{\vec{d}} \cdot \vec{x}^{(1+\lambda)\vec{d}} 
\enspace.
\]

Thus, we get that
\[
g(\vec{x})^{1+\lambda} \geq \sum_{\vec{d}} b_{\vec{d}}^{1+\lambda} \cdot  \vec{x}^{(1+\lambda)\vec{d}}
\geq \theta^\lambda Z^\lambda \sum_{\vec{d} \in A} b_{\vec{d}} \cdot \vec{x}^{(1+\lambda)\vec{d}} 
\geq \theta^\lambda Z^\lambda \cdot \frac{1}{2} \cdot \widehat{g}(\vec{x})
\geq \frac{1}{2n^\tau} \cdot \widehat{g}(\vec{x})
\enspace,
\]
as required.
\end{proof}

Putting together, we can use Lemma~\ref{lemma:sharp_convex} to run Algorithm~\ref{alg:covering} on $\widehat{f}$ to obtain an approximation.

\begin{proof}[Theorem \ref{thm:coveringpoly}(a)]
By Lemma~\ref{lemma:sharp_convex}, when Algorithm~\ref{alg:covering}
is run for the function $\widehat{f}$ with 
\mbox{$\rho := \frac{1}{\tau(1+ \lambda)} \cdot \big(\frac{\tau (1+\lambda)}{\lambda} \big)^{\tau (1+\lambda)}$}
and $\lambda := \frac{1}{\log n}$,
the returned primal vector $\vec{x}^{(m)}$
approximates $\widehat{f}^{\frac{1}{1+\lambda}}$
with competitive ratio $\big(\frac{\tau (1+\lambda)}{\lambda} \big)^{\tau} \leq O(\tau \log n)^\tau$.

Furthermore, by Lemma~\ref{lemma:homo_approx},
$f$ and $\widehat{f}^{\frac{1}{1+\lambda}}$
approximate each other with a multiplicative factor of $2^{\frac{1}{1+\lambda}} \cdot n^{\frac{3 \lambda \tau}{1+\lambda}} = O(1)^\tau$.
Therefore, it follows that the
primal vector $\vec{x}^{(m)}$
approximates $f$ with competitive ratio $O(\tau \log n)^\tau$.
\ignore{

satisfies 
\[
\widehat{f}(\vec{x}^{(m)})^{\frac{1}{\lambda+1}} \leq \bigg(\frac{\tau (1+\lambda)}{\lambda} \bigg)^{\tau} \cdot \widehat{f}(\vec{x}_{\widehat{\OPT}})^{\frac{1}{\lambda+1}} = O(\tau \log n)^{\tau} \cdot \widehat{f}(\vec{x}_{\widehat{\OPT}})^{\frac{1}{\lambda+1}} \enspace.
\]

Finally, observing that $\widehat{f}(\vec{x}_{\widehat{\OPT}}) \leq \widehat{f} (\vec{x}_\OPT)$, we have
\begin{eqnarray*}
C(\vec{x}^{(m)}) & 
= & f(\vec{x}^{(m)}) 
\leq n^{\frac{\lambda\tau}{\lambda+1}} \cdot \widehat{f}(\vec{x}^{(m)})^{\frac{1}{1 + \lambda}} 
\leq (e\tau\ln n)^\tau \cdot   \widehat{f}(\vec{x}_{\widehat{\OPT}})^{\frac{1}{1 + \lambda}}
\leq (e\tau\ln n)^\tau \cdot   \widehat{f}(\vec{x}_{\OPT})^{\frac{1}{1 + \lambda}} \\
& \leq & 2^{\frac{1}{\lambda+1}}(e^5\tau\ln n)^\tau \cdot f(\vec{x}_{\OPT}) 
\leq 2(e^5\tau\ln n)^\tau \cdot \copt \enspace, 
\end{eqnarray*}
as required.}
\end{proof}

\begin{corollary}[$\ell_p$-norm of linear functions]
\label{cor:lp_norm}
Suppose there are $l$ vectors $\vec{c}_k \in \R^n_+$ ($k \in [l]$)
such that each $\vec{c}_k$ has at most $d$ non-zero coordinates.
Define $\vec{\lambda}: \R^n_+ \rightarrow \R^l_+$ such that for $k \in [l]$, $\lambda_k(\vec{x}) := \langle \vec{c}_k, \vec{x} \rangle$.
For some $p \geq 2$, define for $k \in [l]$, 
$f_k(\vec{x}) := \lambda_k(\vec{x})^p$.
Define the cost function $f(\vec{x}) := \sum_{k \in [l]} f_k(\vec{x})$.
Then, there is an $O(p \log d)^p$-competitive online covering algorithm
with respect to $f$.  Equivalently,
the competitive ratio is $O(p \log d)$ with respect to the
$\ell_p$-norm $\| \vec{\lambda}(\vec{x}) \|_p$.
\end{corollary}

\begin{proof}
For each $k \in [l]$,
we consider $f_k$ and $\widehat{f}_k$ as defined in Lemma~\ref{lemma:homo_approx}.
Recall that for each $k \in [l]$,
there is some transformation $\varphi_k: \R_+^n \rightarrow \R_+^n$
of the form $\varphi_k(\vec{x})_i := \alpha_i^\lambda \cdot x_i^{1+\lambda}$
such that $\widehat{f}_k(\vec{x}) = f_k(\varphi_k(\vec{x}))$.
Here, we will eventually set $\lambda := \frac{1}{\log d}$.

For each $k \in [l]$, define $h_k := \widehat{f}_k^{\frac{1}{1+\lambda}}$,
which has the form $h_k(\vec{x}) = \langle \, \widehat{\vec{c}}_k, \vec{x}^{(1+\lambda)\vec{1}} \rangle^{\frac{p}{1+\lambda}}$.  One can check readily that
$\langle \nabla h_k(\vec{x}), \vec{x} \rangle = p \cdot h_k(\vec{x})$
and $\nabla h_k$ is $\lambda$-monotone, since $p\geq 2$ and $0 < \lambda < 1$.
Hence, the function $h := \sum_{k \in [l]} h_k$ also satisfies these properties.

By Lemma~\ref{lemma:sharp_convex}, when Algorithm~\ref{alg:covering}
is run for function $h$, the competitive
ratio is $(\frac{p}{\lambda})^p$ with respect to $h$.

Finally, we observe that for each $k \in [l]$,
Lemma~\ref{lemma:homo_approx} states that
$(\frac{1}{2d^{2\lambda p}})^\frac{1}{1+\lambda} \cdot h_k 
\leq f_k \leq d^\frac{\lambda p}{1+\lambda} \cdot h_k$.
Hence, it follows that $f$ and $h$ approximate each other
with a multiplicative factor
of $O(d^\frac{3\lambda p}{1+\lambda}) = O(1)^p$, for $\lambda := \frac{1}{\log d}$.

Hence, it follows that the competitive ratio is $O(p \log d)^p$ with respect to $f$,
as required.
\end{proof}

\paragraph{General polynomials $f$ with Maximum Degree $\tau$.}
For a general convex multivariate polynomial function $f=\sum_{\vec{d}}c_{\vec{d}}x^{\vec{d}}$
with non-negative coefficients and maximum degree,
the previous method does not work if $f$ is non-homogeneous.
For some $0 < \lambda \leq 1$, we use another approximation $\widetilde{f}=\sum_{\vec{d}}c_{\vec{d}}^{\lambda+1}(x^{\vec{d}})^{\lambda+1}$. 
However, $\widetilde{f}$ might not be convex.
Observe that convexity is used crucially in Lemma~\ref{lemma:convexf}(c).  We tackle this issue
by comparing $\widetilde{f}$ with another convex
function $g := f^{1+\lambda}$.  Let $N$ be the number of monomials in $f$ with positive coefficients.  We denote $\widetilde{\tau} := (1+\lambda) \tau$.

\begin{lemma}[Comparing $\widetilde{f}$ and $g$]
\label{lemma:tf_g}
Suppose $\widetilde{f}$ and $g$ are as defined above.
Then, we have:

\begin{enumerate}[topsep=0cm, itemsep=0cm]
	\item $\widetilde{f}$ is differentiable, monotone, and $f(\vec{0})=0$; 	$\nabla\widetilde{f}$ is $\lambda$-monotone;
 $\widetilde{f}^*$ is monotone;
	
		\item $\langle\nabla\widetilde{f}(\vec{x}),\vec{x}\rangle \le \widetilde{\tau} \cdot f({\vec{x}})$;
		
	
\item $\widetilde{f} \le g \le N^{\lambda} \cdot \widetilde{f}$;

\item $\nabla \widetilde{f} \leq \nabla g$;

\item for all $\vec{z} \in \R_+^n$,
$\widetilde{f}^*(\vec{z}) \leq \frac{1}{N^\lambda} \cdot g^*(N^\lambda \cdot \vec{z})$.
\end{enumerate}
\end{lemma}

\begin{proof}
The first two statements can be verified directly.  
In the third statement, the lower bound is straightforward,
and the upper bound follows from H\"{o}lder's inequality.

We next verify the fourth statement.  
We shall use the inequality that for each $\vec{d}$, $f(\vec{x}) \geq c_{\vec{d}} \cdot \vec{x}^{\vec{d}}$ in the following argument.
For each $i \in [n]$,
\begin{eqnarray*}
\nabla_i g(\vec{x}) & = & (1+\lambda) \cdot f(\vec{x})^\lambda \cdot \nabla_i f(\vec{x}) \\
& = & (1+\lambda) \cdot f(\vec{x})^\lambda \cdot \sum_{\vec{d}: d_i \geq 1} \, c_{\vec{d}} \cdot d_i \cdot \vec{x}^{\vec{d} - \vec{e_i}} \\
& \geq &
\sum_{\vec{d}: d_i \geq 1} \, c_{\vec{d}}^{(1+ \lambda)} \cdot (1+\lambda) d_i \cdot \vec{x}^{(1+\lambda)\vec{d} - \vec{e_i}} \\
& = &  \nabla_i \widetilde{f}(\vec{x}), 
\end{eqnarray*}
as required.

For the fifth statement, we have
$$\widetilde{f}^*(\vec{z}) 
		= \max_{\vec{x} \ge \vec{0}} \big\{ \langle\vec{x},\vec{z}\rangle - \widetilde{f}(\vec{x}) \big\}
		\le \max_{\vec{x} \ge \vec{0}} \big\{ \langle\vec{x},\vec{z}\rangle - \frac{g(\vec{x})}{N^{\lambda}} v\}
		= \frac{1}{N^{\lambda}}\max_{\vec{x} \ge \vec{0}} \big\{ \langle\vec{x},N^{\lambda}\vec{z}\rangle - g(\vec{x}) \big\}
		= \frac{1}{N^{\lambda}} \cdot g^{*}(N^{\lambda}\vec{z}),$$
where the inequality follows from $\widetilde{f} \geq \frac{g}{N^\lambda}$.
\end{proof}

Even though $\widetilde{f}$ might not be convex,
we can still apply Algorithm~\ref{alg:covering} on it.

\begin{lemma}\label{lemma:f_end}
If we apply Algorithm \ref{alg:covering} on $\tilde{f}$, then at the end of round $m$, we have:
\[
\widetilde{f}^{*}(\vec{z}^{(m)}) \le \big( \frac{N^{\lambda}}{\rho\lambda} \big)^{\frac{\widetilde{\tau}}{\widetilde{\tau}-1}} \cdot (\widetilde{\tau}-1) \cdot \widetilde{f}(\vec{x}^{(m)})
\enspace,
\]
where $\rho$ is chosen such that $\frac{N^\lambda}{\rho \lambda} \leq 1$.
\end{lemma}

\begin{proof}
Observe that
Lemma \ref{lemma:boundz2} does not require convexity,
and hence we still have
\[
\vec{z}^{(m)} \le \frac{\nabla\widetilde{f}(\vec{x}^{(m)})}{\rho\lambda}
\leq \frac{\nabla g(\vec{x}^{(m)})}{\rho\lambda}
\enspace.
\] 

Since $\widetilde{f}$ might not be convex,
we compare $\widetilde{f}$ with $g$, and apply
Lemma~\ref{lemma:convexf} to $g$ in the following.
\begin{eqnarray*}
\widetilde{f}^*(\vec{z}^{(m)}) 
& \le & \widetilde{f}^*(\frac{1}{\rho \lambda} \cdot \nabla g(\vec{x}^{(m)})) 
\leq
 \frac{1}{N^{\lambda}} \cdot g^{*} \left(\frac{N^\lambda}{\rho \lambda} \cdot \nabla g(\vec{x}^{(m)}) \right) \\
& \leq & 
\frac{1}{N^{\lambda}} \cdot
\left(\frac{N^{\lambda}}{\rho\lambda} \right)^{\frac{\widetilde{\tau}}{\widetilde{\tau}-1}} \cdot (\widetilde{\tau}-1)
\cdot g(\vec{x}^{(m)})
\leq
\left(\frac{N^{\lambda}}{\rho\lambda} \right)^{\frac{\widetilde{\tau}}{\widetilde{\tau}-1}} \cdot (\widetilde{\tau}-1) \cdot \widetilde{f}(\vec{x}^{(m)}),
\end{eqnarray*}
where $\frac{N^\lambda}{\rho \lambda} \leq 1$ is needed for Lemma~\ref{lemma:convexf}.
\end{proof}
%
%

\begin{proof}[Theorem \ref{thm:coveringpoly}(b)]
When we apply Algorithm~\ref{alg:covering}
to $\widetilde{f}$, 
from Lemma \ref{lemma:yandf}, 
at the end we have
$\sum_{j \in [m]}y_{j} \ge \frac{1}{\rho} \cdot \widetilde{f}(\vec{x}^{(m)})$. 
Combining with Lemma~\ref{lemma:f_end}, we have:
\begin{eqnarray*}
P(\vec{y}) = \sum_{j \in [m]}y_{j} - \widetilde{f}^*(\vec{z}^{(m)})
\ge \frac{1}{\rho} \cdot \widetilde{f}(\vec{x}^{(m)}) -
\left(\frac{N^{\lambda}}{\rho\lambda} \right)^{\frac{\widetilde{\tau}}{\widetilde{\tau}-1}} \cdot (\widetilde{\tau}-1) \cdot \widetilde{f}(\vec{x}^{(m)}).
\end{eqnarray*}

We pick $\rho := (\frac{N^{\lambda}}{\lambda})^{\widetilde{\tau}} \cdot \widetilde{\tau}^{\widetilde{\tau}-1}$.
Then, $\frac{N^\lambda}{\rho \lambda} = \left( \frac{\lambda}{\widetilde{\tau} N^\lambda} \right)^{\widetilde{\tau}-1} \leq 1$,
and we have
$P(\vec{y}) \ge \left( \frac{\lambda}{\widetilde{\tau} N^\lambda} \right)^{\widetilde{\tau}} \cdot \widetilde{f}(\vec{x}^{(m)})$.

We pick $\lambda := \frac{1}{\log N}$.
Hence, the vector $\vec{x}^{(m)}$
is $O(\tau \log N)^\tau$-competitive
with respect to $\widetilde{f}^{\frac{1}{1+\lambda}}$.

Finally, observing that $f$ and $\widetilde{f}^{\frac{1}{1+\lambda}}$
approximate each other with a multiplicative factor of $N^\frac{\lambda}{1 + \lambda} = O(1)$,
we conclude that $\vec{x}^{(m)}$ is also
$O(\tau \log N)^\tau$-competitive with respect to $f$.
%
%
%
%
%
%
%
\end{proof}

\subsection{Primal Feasibility}  
\label{sec:mono_strict}

Observe that in Algorithm~\ref{alg:covering},
in round $k \in [m]$,
each $x_i$ is increased according to the rule 
\mbox{$\frac{d x_i}{d y_k} = \frac{\rho a_{ki} x_i}{\nabla_i f(\vec{x})}$}.
So far we have assumed that if $a_{ki} > 0$, then $x_i$ will be
increased as $y_k$ increases to makes sure the covering constraint can be satisfied (i.e., $\sum_{i = 1}^n a_{ki} x_i \ge 1$).  We now resolve some technical issues concerning the increase of $\vec{x}$.

\begin{itemize}[leftmargin=.5cm, topsep=0cm]
\item What happens if $\nabla_i f(\vec{x}) = 0$?
Indeed, if $f$ is a polynomial such that every $x_i$ occurs with degree at least $2$, then $\nabla_i f(\vec{0}) = 0$.
Suppose $i \in [n]$ is such that $a_{ki} > 0$
and $\nabla_i f(\vec{x}) = 0$.  This means that currently
we can increase $x_i$ such that the rate of change in $f$ is 0.
If there exists some $\epsilon > 0$ such that for all $t \in [0, \epsilon]$, $\nabla_i f(\vec{x} + t \vec{e}_i) = 0$,
then we can increase $x_i$ by $\epsilon$ without changing the value of $f$.  Therefore, we can increase $x_i$ for each such $i$, and assume that for all $i \in [n]$ such that $a_{ki} > 0$,
for all $\epsilon > 0$, $\nabla_i f(\vec{x} + \epsilon \vec{e}_i) > 0$.
This assumption is used in Lemma~\ref{lemma:mono_strict}.
\item If $\vec{x}$ is initialized to $\vec{0}$,
then $\vec{x} \equiv \vec{0}$ is a solution to 
$\frac{d x_i}{d y_k} = \frac{\rho a_{ki} x_i}{\nabla_i f(\vec{x})}$.
One way to resolve this is to initialize $\vec{x} := \vec{L}$ such that each coordinate has some small non-zero value to give each $x_i$ some non-zero momentum to increase.  
However, this introduce an extra additive error term $C(\vec{L})$ as in Theorem~\ref{thm:coveringgeneral}. 
We formally explain this approach in Lemma \ref{lem:mono_strict_general}.

In order to keep a multiplicative competitive ratio while initializing $\vec{x} = \vec{0}$ as in Theorem \ref{thm:coveringpoly}, we show in Lemma~\ref{lemma:mono_strict} that our $\lambda$-monotone condition on $\nabla f$ ensures that $x_i$ can be increased strictly as $y_k$ increases.  This is analogous to the situation that
the differential equation $\frac{d x}{d t} = \sqrt{x}$ with initial condition $x(0) = 0$ has a trivial zero solution,
but what we need is the existence of a non-trivial solution $x(t) = \frac{t^2}{4}$.
\end{itemize}

\begin{lemma}[$y_k$ will stop increasing---general case]
\label{lem:mono_strict_general}
Consider round $k \in [m]$ and suppose for some $i \in [n]$, \mbox{$a_{ki} > 0$}.  
Suppose further that at the beginning of round $k$, the $i$-th coordinate $x_i$ takes some non-zero value $L_i$.  
Then, eventually the constraint $\langle \vec{a}_k, \vec{x} \rangle \geq 1$ will be satisfied, and $y_k$ will stop increasing.
\end{lemma}

\begin{proof}  Observe that for every $i \in [n]$ such that
$a_{ki} > 0$, round $k$ finishes as soon as $x_i$ gets at least $\frac{1}{a_{ki}}$.  Hence, we have an upper bound $\vec{U}$ for the vector $\vec{x}$ throughout round $k$.  In particular,
since $\nabla f$ is monotone, we have $\nabla f(\vec{x}) \leq \nabla f(\vec{U})$.

Therefore, for each $i \in [n]$ such that $a_{ki} > 0$
and $x_i$ is initialized to some $L_i > 0$, we have

\noindent $\frac{d x_i}{d y_k} = \frac{\rho a_{ki} x_i}{\nabla_i f(\vec{x})}
\geq \frac{\rho a_{ki} L_i}{\nabla_i f(\vec{U})}$, which is some positive constant independent of $y_k$. It follows that if $y_k$ keeps on increasing, eventually $x_i$ will reach at least $\frac{1}{a_{ki}}$ and $y_k$ will stop increasing.
\end{proof}


\begin{lemma}[$y_k$ will stop increasing---$\lambda$-monotone case]
\label{lemma:mono_strict}
Consider round $k \in [m]$ and suppose for some $i \in [n]$, \mbox{$a_{ki} > 0$}.  
Suppose further that at the beginning of round $k$, the primal vector takes value $\vec{x}^{(0)}$ such that for all $\epsilon > 0$, $\nabla_i f(\vec{x}^{(0)} + \epsilon \vec{e}_i) > 0$.
Then, eventually the constraint $\langle \vec{a}_k, \vec{x} \rangle \geq 1$ will be satisfied, and $y_k$ will stop increasing.
\end{lemma}

\begin{proof}
Observe that for $i \in [n]$, $x_i$ will be increased
to at most $\frac{1}{a_{ki}}$; if $x_i$ is already at least
$\frac{1}{a_{ki}}$ at the beginning of round $k$,
then round $k$ is finished immediately.
Hence, we can obtain an upper bound $\vec{U}$ for the vector $\vec{x}$
throughout round $k$.

We have $\frac{d x_i}{d y_k} = \frac{\rho a_{ki} x_i}{\nabla_i f(\vec{x})}$.  Since $\nabla_i f(\vec{x}) > 0$ when the $i$th coordinate is increased,  we can write
\[
\textstyle
\rho a_{ki} = \frac{\nabla_i f(\vec{x})}{x_i} \cdot \frac{d x_i}{d y_k} \leq \frac{\nabla_i f(\vec{U})}{U_i^\lambda} \cdot x_i^{\lambda - 1} \cdot \frac{d x_i}{d y_k} \enspace,
\]
where the inequality holds because $\nabla f$ is $\lambda$-monotone.

Integrating with respect to $y_k$ as $y_k$ increases from 0 to $t$,
we have
$\rho a_{ki} t  \leq \frac{\nabla_i f(\vec{U})}{\lambda \cdot U_i^\lambda} \cdot (x_i(t)^\lambda - x_{0i}^\lambda)$.

Hence, $x_i(t) \geq (\frac{\lambda \rho a_{ki} U_i^\lambda}{\nabla_i f(\vec{U})} \cdot t + x_{0i}^\lambda)^{\frac{1}{\lambda}}$
increases strictly as a function of $t$.

Moreover, observing that the last expression tends to $+\infty$ as $t$ tends to $+\infty$, we can conclude that eventually the covering constraint $\sum_{i\in[n]} a_{ki} x_i \geq 1$ will be satisfied.
\end{proof}

\section{Convex Online Packing}
\label{sec:packing}

For the convex online packing problem, we will make explicit assumptions on the packing cost function $f^*$ and exploit them to design and analyze our online algorithms.
Hence, even though we use the same notation, it is more natural to treat $\vec{y}$ as the primal vector and $\vec{x}$ as the dual vector.
Again, we start with some regularity assumptions that we have discussed in the preliminary.

\ignore{
By exchanging the primal and dual, of convex online cover problem, we get the notation of convex online packing problem.

\begin{minipage}[t]{0.45\textwidth}
	\begin{align*}
		\text{(Dual) min } & f(\vec{x})  \\
		\text{subject to: } & A\vec{x} \geq \vec{1} \\
		& \vec{x} \in \R_+^n
	\end{align*}	
\end{minipage}
\hfill\vrule\hfill
\begin{minipage}[t]{0.45\textwidth}
	\begin{align*}
		\text{(Primal) max } & \langle\vec{y},\vec{1}\rangle - f^{*}(\vec{z}) \\
		\text{subject to: } & A^{T}\vec{y} = \vec{z} \\
		& \vec{y} \in \R_+^m \\
		& \vec{z} \in \R_+^n
	\end{align*}
\end{minipage}

For the sake of consistency, we stick to use $C(\vec{x})$ and $P(\vec{y})$.

}


\begin{assumption}
\label{assumption:function_property}
The function $f^*$ is convex, differentiable, and monotone, and $f^*(\vec{0})=0$. 
\end{assumption}

Similar to the covering case, we will further assume that the marginal cost is monotone.
Note that the marginal packing cost function, i.e., $\nabla f^*$, being monotone is not the same as that the marginal covering cost function, i.e., $\nabla f$, being monotone.

\begin{assumption}
The gradient $\nabla f^*$ is monotone.
\end{assumption}

Note that there are instances for which the offline convex packing problem whose objective can be arbitrarily large. 
Consider, for example, a convex packing problem with a linear cost function $f^*(\vec{z}) = \tfrac{1}{2n} \sum_{i \in [n]} z_i$ and the first request is $\vec{a}_1 = \vec{1}$.
Then, by letting $y_1$ to be arbitrarily large, we can get a feasible packing assignment with arbitrarily large objective.
Obviously, such instances are not interesting for practical purposes.
Hence, we will assume that the offline optimal is bounded.

\begin{assumption}
The offline convex packing problem has bounded optimal objective.
\end{assumption}

Our main result of for the online convex packing problem is an online algorithm with the optimal competitive ratio for polynomial cost functions.

\begin{theorem}[Polynomial Cost Functions]
\label{thm:packingpoly}
Suppose the packing cost function $f^*$ is a convex polynomial with non-negative coefficients, zero constant term, and maximum degree $\tau$.
Then, there is an $O(\tau)$-competitive online algorithm for the online convex packing problem.
\end{theorem}

Due to a lower bound result of Huang and Kim~\cite{HuangK15}, the above ratio is asymptotically tight even for the special case of a single resource with a degree-$\tau$ polynomial cost function.

As a key intermediate step for proving Theorem \ref{thm:packingpoly}, we show the following positive results for any cost functions that are, informally speaking, (1) ``at most as convex as'' degree-$\tau$ polynomials and (2) ``at least as convex as'' degree-$\lambda$ polynomials, for some parameters $\tau \ge \lambda > 1$.

\begin{theorem}
\label{thm:packinggeneral}
Suppose the packing cost function $f^*$ satisfies the following conditions:
\begin{enumerate}[leftmargin=1cm, topsep=0cm, itemsep=0cm]
\item There exists some $\tau$ such that for all $\vec{z} \in \R^n_+$, $\langle \nabla f^*(\vec{z}), \vec{z} \rangle \leq \tau \cdot f^*(\vec{z})$.
\item There exists some $\lambda > 1$ such that for all $\rho \geq 1$, for all $\vec{z} \in \R^n_+$,
$\nabla f^*(\rho \vec{z}) \geq \rho^{\lambda-1} \cdot \nabla f^*(\vec{z})$.
\end{enumerate}
Then, there is an $O \big( \tfrac{\tau \lambda}{\lambda - 1} \big)$-competitive online algorithm for the online convex packing problem.
\end{theorem}



Note that if the cost function $f^*$ is a polynomial with maximum degree $\tau$ and minimum degree $\lambda \ge 2$, then Theorem \ref{thm:packinggeneral} implies that there is an $O( \tau )$-competitive online algorithm.
To further prove Theorem \ref{thm:packingpoly} for polynomials with linear terms, we will explain in Section \ref{sec:packingpoly} how to handle the linear terms and reduce the problem to polynomials of degree at least $2$.


\subsection{Online Convex Packing Algorithm for Theorem \ref{thm:packinggeneral}}

The details are described in Algorithm~\ref{alg:packing}.

\begin{algorithm}[H]
	\label{alg:packing}
	\SetAlgoLined
	\textbf{Initialize:} $\vec{x}:= \vec{z} :=\vec{0}$\;
	\While{constraint vector $\vec{a}_k = (a_{k1}, \ldots, a_{kn})$ arrives in round $k$}{
		Set $y_{k} :=0$; \\
		\While{$\sum_{i=1}^{n}a_{ki}x_{i}<1$}{
			Continuously increase $y_k$: \\
			Simultaneously for each $i \in [n]$, increase
			$z_i$ at rate $\frac{d z}{d y_k} = a_{ki}$, and \\
			increase $\vec{x}$ according to $\vec{x}=\nabla f^{*}(\rho\vec{z})$\;
		}
	}
	\caption{Convex online packing}
\end{algorithm}

Here, the vector $\vec{x}$ plays an auxiliary role
and is initialized to $\vec{0}$. Throughout the algorithm,
we maintain the invariant $\vec{z} = A^T \vec{y}$
and $\vec{x} = \nabla f^*(\rho \vec{z})$ for some 
parameter $\rho > 1$ to be determined later.
In round $k \in [m]$, the vector $\vec{a}_k = (a_{k1}, a_{k2},
\ldots, a_{kn})$ is given. The variable $y_k$ is initialized
to 0, and is continuously increased while $\sum_{i \in [n]} a_{ki} x_i < 1$.  To maintain $\vec{z} = A^T \vec{y}$,
for each $i \in [n]$, $z_i$ is increased at rate $\frac{d z}{d y_k} = a_{ki}$.  As the coordinates of $\vec{z}$ are increased,
the vector $\vec{x}$ is increased according to the invariant $\vec{x} = \nabla f^*(\rho \vec{z})$.
Since $\nabla f^*$ is monotone,
as $y_k$ increases, both $\vec{z}$ and $\vec{z}$ increase
monotonically.  
We show in Lemma~\ref{lemma:packing_feasible} that unless
the offline packing problem is unbounded,
eventually $\sum_{i \in [n]} a_{ki} x_i$ reaches 1, at which moment
$y_k$ stops increasing and round $k$ finishes.

Observe that the coordinates of $\vec{x}$ are increased monotonically throughout the algorithm.
Below we show that at the end of the
process, the constraints $\sum_{i \in [n]} a_{ji} x_i \geq 1$ are satisfied for all $j \in [m]$.  
Hence, the vector $\vec{x}$ is feasible for the covering problem.

In the rest of the section, for $k \in [m]$, we let $\vec{z}^{(k)}$ denote the vector $\vec{z}$ at the end of round $k$, where $\vec{z}^{(0)} := \vec{0}$.

\begin{lemma}[Dual Feasibility]
\label{lemma:packing_feasible}
Recall our assumption that the offline optimal packing objective is bounded.
Then, in each round $k\in[m]$, eventually we have $\sum_{i \in [n]} a_{ki} x_i \geq 1$, and $y_k$ will stop increasing.
\end{lemma}

\begin{proof}
During round $k \in [m]$,
the algorithm increases $y_k$ only when $\sum_{i=1}^{n}a_{ki}x_{i}<1$.
Therefore, recalling $\vec{z} = A^T \vec{y}$, when the algorithm increases $y_k$, it also increases each $z_i$ at rate $\frac{d z_i}{d y_k} = a_{ki}$.

Hence, we have
\begin{align*}
\tfrac{\partial P(\vec{y})}{\partial y_k} = 1 - \langle \vec{a}_k, \nabla f^*(\vec{z}) \rangle
& \geq 1 - \tfrac{1}{\rho^{\lambda - 1}} \cdot \langle \vec{a}_k, \nabla f^*(\rho \vec{z}) \rangle \\
& = 1 - \tfrac{1}{\rho^{\lambda - 1}} \cdot \langle \vec{a}_k, \vec{x} \rangle \geq 1 - \tfrac{1}{\rho^{\lambda - 1}},
\end{align*}
where the first inequality follows from the assumption $\nabla f^*(\rho \vec{z}) \geq \rho^{\lambda - 1} \cdot \nabla f^*(\vec{z})$,
and the last inequality follows because $\langle \vec{a}_k, \vec{x} \rangle < 1$ when $y_k$ is increased.

Therefore, suppose for contrary that $\langle \vec{a}_k, \vec{x} \rangle$ never reaches 1,
then the objective function $P(\vec{y})$ increases at least at some
positive rate $1 - \frac{1}{\rho^{\lambda - 1}}$ (recalling $\rho > 1$ and $\lambda \geq 2$) as $y_k$ increases, which means
the offline packing problem is unbounded, contradicting our assumption.
\end{proof}

\ignore{

In the algorithm, $\rho>1$ is a variable whose value will be determined later.

\begin{lemma}
	At the end of the algorithm, both $C(\vec{x})$ and $P(\vec{y})$ are feasible.
\end{lemma}

\begin{proof}
	$P(\vec{y})$ is always feasible during the algorithm because
	\begin{enumerate}
		\item $\vec{y}$ can only be increased during the whole process, therefore $\vec{y} \in \R_+^m$.
		\item In the algorithm, we always maintain $z_{i}=\sum_{j=1}^{m}y_{j}a_{ji}$. Since all the entries of $A$ are nonnegative, it always holds that $A^{T}\vec{y} = \vec{z}$, and $\vec{z} \in \R_+^n$.
	\end{enumerate}
	Also, after a new constraint comes, it only halts when $\sum_{i=1}^{n}a_{ki}x_{i}\geq1$. Since $f^{*}$ is strictly convex and monotone, $f*(\rho\vec{z})$ is unbounded. Therefore, we can always reach the condition that $\sum_{i=1}^{n}a_{ki}x_{i}\geq1$.
\end{proof}
}

\subsection{Competitive Analysis for Theorem \ref{thm:packinggeneral}}

\begin{lemma}[Bounding Increase in $\vec{y}$]
\label{lemma:yinc}
For $k \in [m]$, let $\vec{z}^{(k)}$
denote the vector $\vec{z}$ at the end of round $k$,
where $\vec{z}^{(0)} := \vec{0}$.
Then, 
at the end of round $k$ when $y_k$ stops increasing (by Lemma~\ref{lemma:packing_feasible})
\[
y_k \geq \tfrac{1}{\rho} \cdot \big( f^*(\rho \vec{z}^{(k)}) - f^*(\vec{z}^{(k-1)}) \big) \enspace.
\]
In particular, since $f^*(\vec{0})  = 0$, this implies that at the end of the algorithm,
\[
\sum_{k \in [m]} y_k \geq \tfrac{1}{\rho} \cdot f^*(\rho \vec{z}^{(m)}) \enspace.
\]
\end{lemma}

\begin{proof} 
Recall again that $y_k$ increases only when $\langle \vec{a}_k, \vec{x} \rangle < 1$, we have
\[
1 \geq \sum_{i \in [n]} a_{ki} x_i = \sum_{i \in [n]} x_i \cdot \frac{d z_i}{d y_k} \enspace.
\]
Hence, integrating this with respect to $y_k$
throughout round $k$, and observing that $\vec{x} = \nabla f^*(\rho \vec{z})$, we have
\[
y_k \geq \int_{\vec{z} = \vec{z}^{(k-1)}}^{\vec{z}^{(k)}} 
\langle \nabla f^*(\rho \vec{z}), d \vec{z} \rangle
= \frac{1}{\rho} \cdot (f^*(\rho \vec{z}^{(k)}) - f^*(\vec{z}^{(k-1)})) \enspace,
\]
where the last equality comes from the fundamental theorem of calculus for path integrals of vector fields.
\end{proof}


\begin{proof}[Theorem \ref{thm:packinggeneral}]
Suppose $\vec{z}^{(m)} = A^T \vec{y}$ is the vector
at the end of the algorithm, and $\vec{x}^{(m)} = \nabla f^*(\rho \vec{z}^{(m)})$.
Then, we have
\[
C(\vec{x}^{(m)}) = f(\nabla f^*(\rho \vec{z}^{(m)})) \leq (\tau - 1) \cdot f^*(\rho \vec{z}^{(m)}) \enspace,
\]
where the inequality follows from applying Lemma~\ref{lemma:convexf}(c) with the roles of $f$ and $f^*$ reversed.

On the other hand, 
\begin{equation}
\label{eq:packing1}
\textstyle
P(\vec{y}) = \sum_{k \in [m]} y_k - f^*(\vec{z}^{(m)}) \geq \frac{1}{\rho} \cdot f^*(\rho \vec{z}^{(m)}) - f^*(\vec{z}^{(m)}) \enspace.
\end{equation}

Hence, it follows that
\begin{equation}
\label{eq:calc2}
\frac{C(\vec{x})}{P(\vec{y})} \leq 
\frac{(\tau - 1) \cdot f^*(\rho \vec{z}^{(m)})}{\frac{1}{\rho} \cdot f^*(\rho \vec{z}^{(m)}) - f^*(\vec{z}^{(m)})}
\leq 
\frac{(\tau - 1) \cdot \rho^\lambda \cdot f^*(\vec{z}^{(m)})}{\frac{1}{\rho} \cdot \rho^\lambda \cdot f^*( \vec{z}^{(m)}) - f^*(\vec{z}^{(m)})} 
= (\tau -1 )\cdot \frac{\rho^{\lambda}}{\rho^{\lambda-1} - 1} \enspace,
\end{equation}
where the penultimate inequality follows because 
the assumption $\nabla f^*(\rho \vec{z}) \geq \rho^{\lambda - 1} \cdot \nabla f^*(\vec{z})$ implies
that
$f^*(\rho \vec{z}^{(m)}) = 
\rho \int_{\vec{z} = \vec{0}}^{\vec{z}^{(m)}}
\langle \nabla f^*(\rho \vec{z}), d \vec{z} \rangle
\geq \rho^\lambda \int_{\vec{z} = \vec{0}}^{\vec{z}^{(m)}}
\langle \nabla f^*(\vec{z}), d \vec{z} \rangle = \rho^\lambda \cdot f^*(\vec{z}^{(m)}))$,
and the function $t \mapsto \frac{t}{\frac{t}{\rho} - f^*(\vec{z}^{(m)})}$ is decreasing.

Choosing $\rho := \lambda^{\frac{1}{\lambda-1}}$
and observing that $\vec{x}^{(m)}$ is feasible for the covering problem,
we have 
\[
\frac{\popt}{P(\vec{y})} \leq \frac{\copt}{P(\vec{y})}
\leq \frac{C(\vec{x}^{(m)})}{P(\vec{y})} \leq (\tau-1) \cdot \frac{\lambda^{\frac{\lambda}{\lambda-1}}}{\lambda - 1} \enspace.
\]
The result then follows because $\lambda^{\frac{1}{\lambda - 1}} = \big(1 + (\lambda - 1) \big)^{\frac{1}{\lambda - 1}} \le e$.
\end{proof}

\subsection{Proof of Theorem \ref{thm:packingpoly}}
\label{sec:packingpoly}

\paragraph{Polynomial with Linear Terms.}
Observe that if a polynomial $f^*$ has linear terms, then for any $\lambda > 1$, it does not satisfy the condition that $\nabla f^*(\rho \vec{z}) \geq \rho^{\lambda - 1} \cdot \nabla f^*(\vec{z})$ for any $\vec{z} \in \R_{+}^n$ and any $\rho > 1$. 
We write $f^*(\vec{z}) = \langle \vec{c}, \vec{z} \rangle + {\widehat{f}}^*(\vec{z})$, where $\vec{c} = \nabla f^*(\vec{0})$ and ${\widehat{f}}^*(\vec{z})$ contains terms with degree at least $2$.  
Therefore, for all $\rho > 1$, $\nabla {\widehat{f}}^*(\rho \vec{z}) \geq \rho \cdot \nabla {\widehat{f}}^*(\vec{z})$.

Then, the objective function becomes
\[
P(\vec{y}) = \langle \vec{1}, \vec{y} \rangle - f^*(A^T y) = \langle \vec{1} - A \vec{c}, \vec{y} \rangle - {\widehat{f}}^*(A^T y) \enspace.
\]

Moreover,
the corresponding covering problem becomes
$\min_{\vec{x} \geq 0} \widehat{f}(\vec{x})$ subject to $A \vec{x} \geq \vec{1} - A \vec{c}$.  In other words,
in round $k \in [m]$, as the vector $\vec{a}_k$ arrives,
the covering constraint becomes $\langle \vec{a}_k, \vec{x} \rangle \geq b_k$, where $b_k := 1 -  \langle \vec{a}_k, \vec{c} \rangle$.
If $b_k \leq 0$, then the constraint is automatically satisfied, and
we set $y_k = 0$ such that round $k$ finishes immediately.  Otherwise, we can run Algorithm ~\ref{alg:packing} using the function ${\widehat{f}}^*$ and the constraint vector $\frac{\vec{a}_k}{b_k}$ in round $k$.

Next, we present a proof of Theorem \ref{thm:packingpoly} based on the above discussion.


\begin{proof}[Theorem \ref{thm:packingpoly}]
As discussed above, we write $\vec{c} := \nabla f^*(\vec{0})$ and ${\widehat{f}}^*(\vec{x}) := f^*(\vec{x}) - \langle \vec{c}, \vec{x} \rangle$ as the convex polynomial containing the terms of $f^*$ with degree at least $\lambda = 2$.
(Observe that if $f^*$ contains only linear terms, then the problem is trivial because the objective is unbounded when there is some round $k \in [m]$ such that $b_k > 0$.)


For ease of exposition, we can assume
that for each $k \in [m]$, $b_k := 1 - \langle \vec{a}_k, \vec{c} \rangle > 0$.  Otherwise, we can essentially ignore the variable $y_k$ by setting it to 0.  We denote $B$ as the diagonal matrix whose $(k,k)$-th entry is $b_k$.  By writing $\vec{w} := B \vec{y}$,
the objective function can be expressed in terms of $\vec{w}$
as $\widehat{P}(\vec{w}) := \langle \vec{1}, \vec{w} \rangle - {\widehat{f}}^*((B^{-1}A)^T \vec{w})$.

Hence, we can run Algorithm~\ref{alg:packing} using function ${\widehat{f}}^*$
such that in round $k \in [m]$, when the vector $\vec{a}_k$ arrives,
we can transform it by dividing each coordinate by $b_k$ before passing it to the algorithm.

By Theorem~\ref{thm:packinggeneral}, using $\lambda = 2$, it follows that the algorithm has competitive ratio $O(\tau)$.
\end{proof}

\section{Online Combinatorial Auction with Non-separable Production Cost}
\label{sec:auction}

In this section, we explain how to extend our algorithm for the online convex packing problem to a more general problem known as online combinatorial auction with production cost.

In an online combinatorial auction, there is a seller with $n$ types of items (i.e., resources) that are known upfront and a convex production cost function $g = f^* : \R_{+}^n \rightarrow \R_{+}$.
Producing $z_i$ units of resource $i$ for all $i \in [n]$ incurs a production cost of $f^*(\vec{z})$.
There are $m$ buyers (i.e., requests) that arrive online.
Each buyer $j$ is associated with a value function\footnote{For efficiency issues, we might need to assume that
$v_j$ is supermodular (i.e., for subsets $A$ and $B$,
$v_j(A) + v_j(B) \leq v_j(A \cap B) + v_j(A \cup B)$), but
otherwise we do not need further assumptions.
} $v_j : 2^{[n]} \rightarrow \R_{+}$ such that $v_j(S)$ is buyer $j$'s value for getting a subset of items $S \subseteq [n]$.
On the arrival of a buyer, the seller must choose a subset of items $S_j$ to allocate to the buyer immediately without any information of future buyers.
The objective is to maximize the social welfare, which is defined as total value of buyers, $\sum_{j \in [m]} v_j(S_j)$, minus the production cost $f^*(\vec{z})$, where
$z_i$ is the number of $j$'s such that $i \in S_j$.

We assume that $f^*$ have the same properties as in Section~\ref{sec:packing}, i.e., $f^*$ is convex and differentiable,
and both $f^*$ and $\nabla f^*$ are monotone, and $f^*(\vec{0}) = 0$.  In addition,
we shall consider a couple more technical assumptions on $f^*$, which are true for most interesting functions such as polynomials.

\begin{theorem}
\label{th:poly_auction}
Suppose the cost function $f^*$ is a convex polynomial with non-negative coefficients, zero constant term, and maximum degree $\tau$.
Then, there is an $O(\tau)$-competitive online algorithm running for the online combinatorial auction problem with production cost $f^*$.
\end{theorem}

%

\ignore{
Since the production cost $f^*$ is monotone, we can also assume that the value function $v_j$ is also monotone in the sense that if there are $A \subsetneq B$ such that $v_j(A) > v_j(B)$, then subset $B$ will never be chosen over $A$.
}

Consider the following standard convex program relaxation of combinatorial auction with production costs and its Fenchel dual program:
\begin{align*}
\max_{\vec{y} \ge 0} & \textstyle \quad P(\vec{y}) \defeq \sum_{j \in [m]} \sum_{S \subseteq [n]} v_j(S) \cdot y_{jS} - f^*(\vec{z}) \quad \text{\rm s.t.} \\
\forall i \in [n] : & \textstyle \quad \sum_{j \in [m]} \sum_{S \ni i} y_{jS} = z_i \\
\forall j \in [m] : & \textstyle \quad \sum_{S \subseteq [n]} y_{jS} \le 1 \\
& \\
\min_{\vec{x}, \vec{u} \ge 0} & \textstyle \quad C(\vec{x}, \vec{u}) \defeq f(\vec{x}) + \sum_{j \in [m]} u_j \quad \text{\rm s.t.} \\
\forall j \in [m], \forall S \subseteq [n] : & \textstyle \quad \sum_{i \in S} x_i + u_j \ge v_j(S) 
\end{align*}

Observe that in the packing problem,
the objective function
can be succinctly expressed as $P(\vec{y}) := \langle \vec{v}, \vec{y} \rangle - f^*(A^T \vec{y})$,
where the coordinates of $\vec{v}$ are indexed by $[m] \times 2^{[n]}$,
and $A$ is the $\{0,1\}$-matrix whose rows are indexed by 
$[m] \times 2^{[n]}$ and columns are indexed by $[n]$ 
such that for $j \in [m]$, $S \subseteq [n]$ and $i \in [n]$, the $(jS,i)$-th entry is 1 \emph{iff} $i \in S$.
We also denote $x_S := \sum_{i \in S} x_i$.

In this paper, we present an online algorithm that solves the above convex program fractionally.
In each round $k \in [m]$, the value function $v_k$ arrives,
and the algorithm irrevocably chooses non-negative values for $y_{kS}$ for all $S \subseteq [n]$  such that $\sum_{S \subseteq [n]} y_{kS} \leq 1$.  Observe that the algorithm knows the rows of $A$ in advance, although it may not know the number $m$ of rounds.
Translating fractional algorithms into integral one is relatively straightforward and readers are referred to Huang and Kim~\cite{HuangK15} for details.
The algorithm is given in Algorithm \ref{alg:auction}.

\begin{algorithm}[h]
		\SetAlgoLined
	\textbf{Initialize:} $\vec{x} := \vec{z} := \vec{0}$\;
	\While{buyer $k$ with value function $v_k$ arrives in round $k$}{
		Set $\vec{y}_{k} := \vec{0}$ and $u_{k} := 0$; \\
		
		\While{$t$ increases continuously from $0$ to $1$}{
			
			Define $r(t) := \max_{A \subseteq [n]} \gamma_A(t)$,
			where $\gamma_A(t) := v_k(A) - \sum_{i \in A} x_i(t)$.
			
			Pick any $S \subseteq [n]$ that attains $r(t)$.
			
			\While{$\gamma_S(t) \geq r(t) - \ve$}{
				Increase $t$ continuously:
				
				\begin{enumerate}[leftmargin=0.7cm, topsep=0cm, itemsep=-0.1cm]
				\item Increase $y_{kS}$ at rate $\frac{d y_{kS}}{dt} = 1$.
				\item To maintain the invariant $\vec{z} = A^T \vec{y}$,
				
			        for each $i \in S$, $z_i$ is increased at rate $\frac{d z_i}{d t} =1$.
			\item As $\vec{z}$ is increased, we maintain $\vec{x} := \nabla f^*(\rho \vec{z})$ for some parameter $\rho > 1$.
			\item Increase $u_k$ at rate $\frac{d u_k}{d t}	= r(t)$. (Note that $u_k$ is for analysis only.)
			\end{enumerate}
					
			}

		%
			%
			%
			%

		}
	}
	\caption{Online combinatorial auction with production cost}
	\label{alg:auction}
\end{algorithm}

\noindent \textbf{Explanation of Algorithm.}
In each round $k \in [m]$, 
the variables are changed continuously as functions of some time parameter $t \in [0,1]$
such that $t=0$ corresponds to the beginning of round $k$,
and round $k$ finishes when $t$ reaches 1.
Initially, $\vec{y}_k(0) := \vec{0}$ and $u_k(0) := 0$.
  The following invariants
are maintained.

\noindent \textbf{Invariant 1:} The sum $\sum_{S \subseteq [n]} y_{kS}(t)$ increases at rate 1 with respect to $t$.

In fact, the algorithm ensures that at any time $t$, there is exactly one $S \subseteq [n]$ such that $y_{kS}$ is increased at rate 1.
This ensures that when $t$ reaches 1, the sum
$\sum_{S \subseteq [n]} y_{kS}(t)$ is 1 to maintain
the feasibility of $\vec{y}$.  To decide which $y_{kS}$'s to increase at time $t$, a parameter is defined $r(t) := \max_{A \subseteq [n]} \gamma_A(t)$, where
$\gamma_A(t) := v_k(A) - x_A(t)$.
Intuitively,
the sets $S$ that attain $r(t)$ are the most worthwhile to be selected.
One technical issue is that whether $r(t)$ can be computed efficiently.  If the set function $v_k$ is supermodular,
then $r(t)$ can be computed efficiently.


\noindent \textbf{Invariant 2:} A parameter $\ve > 0$ is chosen such that a variable $y_{kS}$ is increased at time $t$ only if $\gamma_S(t) \geq r(t) - \ve$.

We shall see that this invariant is used to bound the competitive ratio.
One might attempt to define Invariant 2 with $\ve = 0$.
The problem is that when $y_{kS}$ is increased, the vectors $\vec{z} = A^T \vec{y}$
and $x := \nabla f^*(\rho \vec{z})$ will also be increased such that
the set $S$ might no longer satisfy $\gamma_S(t) \geq r(t)$, even
when $t$ is increased infinitesimally.  
To choose the value of $\ve$ and facilitate the implementation of the algorithm, we place some technical assumptions on $f^*$,
which are true for interesting functions.

\begin{compactitem}
\item The gradient $\nabla f^*$ is locally Lipschitz with respect to the $\ell_1$-norm, i.e.,
for all $\vec{z}$, for all $R > 0$, there exists some $L$ such that
$\| \vec{z} - \vec{a} \|_1$, $\| \vec{z} - \vec{b} \|_1$  $\leq R$ implies that
$\| \nabla f^*(\vec{a}) - \nabla f^*(\vec{b}) \|_1 \leq L \cdot \|\vec{a} - \vec{b}\|_1$.  

\item For all $R > 0$, the infimum  $\inf_{\|\vec{z}\|_1 \geq R} \frac{f^*(\vec{z})}{\|\vec{z}\|_1}$ is positive.  This assumption means that the production cost cannot be zero as long as some resource is being used, and must grow at least proportionately as more resources are used.
\end{compactitem}

We shall see that it is sufficient to choose 
$\ve := \frac{1}{10} \inf_{\|\vec{z}\|_1 \geq R} \frac{f^*(\vec{z})}{\|\vec{z}\|_1}  > 0$,
where $R$ depends on the first value function $v_1$ and the local Lipschitz
constant of $\nabla f^*$ around $\vec{0}$.

\noindent \textbf{Continuous vs Discrete Increments.}
Observe that during round $k$, the algorithm needs
to change the $y_{kS}$ to increase when $\gamma_S(t) < r(t) - \ve$;
this means that $\gamma_S(t)$ must have decreased by at least $\ve$ from the time
$S$ is selected.
Moreover, observe that $r(t) \geq \gamma_\emptyset(t) \geq 0$.
Hence, it follows that there can be at most $2^n \cdot \frac{r(0)}{\ve}$ changes of $S$ before $t$ reaches 1.



We show that
the local Lipschitz property of $\nabla f^*$ implies that
instead of monitoring $r(t)$ continuously, the algorithm can be implemented by discrete increments, even though it is more convenient to analyze it continuously.

Notice that in round $k$, the $\ell_1$-norm of the vector $\vec{z} = A^T \vec{y}$ can increase by at most $n$, which happens if the complete set $[n]$ is chosen throughout.  Hence,
it follows that $\| \rho \vec{z} \|_1$ can change by at most $\rho n$ according to the $\ell_1$-distance.  Let $L$ be the local Lipschitz constant for $\nabla f^*$ in this vicinity of $\rho \vec{z}$ during round $k$.  Observe that the mapping $t \mapsto \vec{x}(t) := \nabla f^*(\rho \vec{z}(t))$ is $L \rho n$-Lipschitz with respect to the $\ell_1$-norm.

Therefore, if $t$ is increased by $\delta := \frac{\ve}{L \rho n}$,
$\| \vec{x}(t) \|_1$ can increase by at most $\ve$.
Hence, it follows that if $\gamma_S(t_0) = r(t_0)$, then for all $t \in [t_0, t_0 + \delta]$,
we have $\gamma_S(t) \geq \gamma_S(t_0) - \ve \geq r(t) - \ve$,
since $r(t)$ is non-increasing.
As a result, we can increase $t$ at increments of $\delta$,
and compute $r(t)$ and change subsets $S$ for only $\frac{1}{\delta}$ times before $t$ reaches 1, and round $k$ finishes.

\ignore{

Hence, we need to increase $y_{kS}$ for $S \in \mathcal{S}_t$ together.
One can see that increasing each of them at the same rate $\frac{1}{|\mathcal{S}_t|}$ to maintain Invariant 1 might not work,
because for each $S \in \mathcal{S}_t$,
the quantity $v_k(S) - x_S(t)$ decreases at rate $\frac{d x_S}{d t}$,
which might not be the same if we increase the $y_{kS}$'s uniformly
at the same rate. (Here, we also observe that
if $S_1 \subsetneq S_2$, then $\frac{d x_{S_1}}{d t} \leq \frac{d x_{S_2}}{d t}$, and this is why we can consider only minimal subsets in $\mathcal{S}_t$.)

Observing that 
$\frac{d x_i}{d t} = \rho \sum_{j \in [n]} \frac{\partial \nabla_i f^*(\rho \vec{z})}{d z_j} \sum_{S \in \mathcal{S}_t: j \in S} \frac{d y_{kS}}{d t}$ is a linear combination
of the $\frac{d y_{kS}}{d t}$'s,
one can choose $\frac{d y_{kS}}{d t} := \beta_S(t)$ for all $S \in \mathcal{S}_t$ such that $\sum_{S \in \mathcal{S}_t} \beta_S(t) = 1$,
and for all $S \in \mathcal{S}_t$, $x_S(t)$ increases at the same rate.
}

\noindent \textbf{Feasibility of Covering Problem.}
Observe that to maintain the feasibility of $(\vec{x}, \vec{u})$,
at the end of round $k$,
one can simply set $u_k := \max_{S \subseteq [n]} v_k(S) - x_S$.
However, to facilitate the competitive analysis,
we increase $u_k$ at rate $\frac{d u_k}{d t} = r(t)$ throughout round $k$.  The next lemma shows that this also maintains
the feasibility of $(\vec{x}, \vec{u})$.

\begin{lemma}[Feasibility of $(\vec{x}, \vec{u})$]
\label{lemma:feas_xu}
After each round $k \in [m]$,
for all $S \subseteq [n]$,
$u_k \geq v_k(S) - x_S$.
\end{lemma}

\begin{proof}
Recall that the time parameter $t \in [0,1]$ denotes the beginning
of round $k$ with $t=0$ and the end of round $k$ with $t=1$.  Observing that $\vec{x}$ is increased monotonically,
for all $t \in [0,1]$, $x_S(t) \leq x_S(1)$.

Hence, we have $\frac{d u_k}{d t} = r(t) \geq v_k(S) - x_S(t) \geq v_k(S) - x_S(1)$.  Integrating with respect to $t$ from 0 to 1 gives the result.
\end{proof}

\noindent \textbf{Competitive Analysis.}
The analysis is along the same lines as that of the convex packing problem.
The main difference lies in how we bound the increase in $\vec{y}$.
Concretely, we will use the following lemma, which is an analogue of Lemma \ref{lemma:yinc} in the convex packing problem. 

\begin{lemma}[Bounding Increase in $\vec{y}$]
\label{lemma:yinc_auction}
For $k \in [m]$, let $\vec{z}^{(k)}$
denote the vector $\vec{z}$ at the end of round $k$,
where $\vec{z}^{(0)} := \vec{0}$.
Then, 
at the end of round $k$,
\[
\sum_{S \subseteq [n]} v_k(S) \cdot y_{kS} \geq u_k + \tfrac{1}{\rho} \cdot \big( f^*(\rho \vec{z}^{(k)}) - f^*(\rho \vec{z}^{(k-1)}) \big) - \ve R_k \enspace,
\]
where $R_k$ is the amount of time in $[0,1]$ during round $k$ in which
some non-empty set $S$ is chosen to increase the variable $y_{kS}$.

In particular, since $f^*(\vec{0})  = 0$, this implies that at the end of the algorithm,
\[
\langle \vec{v}, \vec{y} \rangle = \sum_{j \in [m]} \sum_{S \subseteq [n]} v_j(S) \cdot y_{jS} \geq \sum_{j \in [m]} u_j + \frac{1}{\rho} \cdot f^* \big( \rho \vec{z}^{(m)} \big) - \ve \sum_{j \in [m]} R_j \enspace.
\]
\end{lemma}

\begin{proof} 
Recall that by Invariant 2, at time $t \in [0,1]$,
$y_{kS}$ increases only if $\gamma_S(t) \geq r(t) - \ve$.
On the other hand, observe that if the empty set $S = \emptyset$ is chosen
such that $y_{k\emptyset}$ is increased, both $\vec{z}$ and $\vec{x}$ remains the same.  Hence, the invariant $\gamma_\emptyset(t) = r(t)$ is actually maintained with no error term.  

We define the indicator function $\chi : [0,1] \rightarrow \{0,1\}$ such that
$\chi(t) = 1$ \emph{iff} a non-empty set $S$ is chosen to increase $y_{kS}$ at time $t$.  Then, we have $\gamma_S(t) \geq r(t) - \ve \cdot \chi(t)$

Recall that $\frac{d u_k}{d t} = r(t)$.  Hence, if $y_{kS}$ is increased at time $t$, we have

\[ v_k(S) \geq \frac{d u_k}{d t} + x_S(t) - \ve \cdot \chi(t).\]

By Invariant 1, $\sum_{S' \subseteq [n]} \frac{d y_{kS'}}{d t} = 1$.
Observe that for other $S' \neq S$, $\frac{d y_{kS'}}{d t} = 0$.  Hence, we can multiply the above equation by $\frac{d y_{kS}}{d t} = 1$,  and include
the zero terms in the sum for $S' \neq S$ to obtain the following.

\[ \sum_{S \subseteq [n]} v_k(S) \cdot \frac{d y_{kS}}{d t}
\geq \frac{d u_k}{d t} + \sum_{S \subseteq [n]} \sum_{i \in S} x_i \cdot \frac{d y_{kS}}{d t} - \ve \cdot \chi(t) = \frac{d u_k}{d t} +  \sum_{i \in [n]} x_i \cdot \frac{d z_i}{d t} - \ve \cdot \chi(t),\]

\noindent where the last equality follows from
interchanging the order of summation, and $\frac{d z_i}{d t} = \sum_{S \subseteq [n]: i \in S} \frac{d y_{kS}}{d t}$.

Observe that $\vec{x} = \nabla f^*(\rho \vec{z})$.
Hence, integrating with respect to $t$ from $0$ to $1$,
the vector $\vec{z}$ increases from $\vec{z}^{(k-1)}$
to $\vec{z}^{(k)}$, and we have:
\[
\sum_{S \in [n]} v_k(S) \cdot y_{kS} \geq u_k + \int_{\vec{z} = \vec{z}^{(k-1)}}^{\vec{z}^{(k)}} 
\langle \nabla f^*(\rho \vec{z}), d \vec{z} \rangle - \ve \cdot R_k
= u_k + \frac{1}{\rho} \cdot (f^*(\rho \vec{z}^{(k)}) - f^*(\rho \vec{z}^{(k-1)})) - \ve \cdot R_k \enspace,
\]
where $R_k := \int_0^1 \chi(t) dt$ is the amount of time
in which a non-empty set $S$ is chosen to increase $y_{kS}$,
and the last equality comes from the fundamental theorem of calculus for path integrals of vector fields.
\end{proof}

\begin{proof} [Theorem~\ref{th:poly_auction}]
We first use the trick in Section~\ref{sec:packingpoly}
to absorb the linear terms of $f^*$ into $\langle \vec{v}, \vec{y} \rangle$ in the objective function $P(\vec{y}) = \langle \vec{v}, \vec{y} \rangle - f^*(A^T y)$.  Hence, we can assume that
each term in $f^*$ has degree at least 2.

By Lemma \ref{lemma:yinc_auction}, after round $m$, we have
\[
\textstyle
P(\vec{y}) = \langle \vec{v}, \vec{y} \rangle - f^*(\vec{z}^{(m)}) \ge \langle \vec{1}, \vec{u} \rangle + \frac{1}{\rho} \cdot f^* \big( \rho \vec{z}^{(m)} \big) - f^*(\vec{z}^{(m)}) - \ve R,
\]
where $R := \sum_{j \in [m]} R_j$.

Recalling that $\vec{x} = \nabla f^*(\rho \vec{z}^{(m)})$,
$C(\vec{x}, \vec{u}) \defeq \langle \vec{1}, \vec{u} \rangle +  f(\nabla f^*(\rho \vec{z}^{(m)})) \leq \langle \vec{1}, \vec{u} \rangle 
+ (\tau - 1) \cdot f^*(\rho \vec{z}^{(m)})$,
where the last inequality follows from Lemma~\ref{lemma:convexf}(c),
since the polynomial $f^*$ is has maximum degree $\tau$.

Hence, similar to inequality (\ref{eq:calc2}) in the proof of
Theorem~\ref{thm:packinggeneral},
we have:

\[\textstyle
\frac{C(\vec{x}, \vec{u})}{P(\vec{y})} \leq 
\frac{\langle \vec{1}, \vec{u} \rangle + (\tau - 1) \cdot f^*(\rho \vec{z}^{(m)})}{\langle \vec{1}, \vec{u} \rangle + \frac{1}{\rho} \cdot f^*(\rho \vec{z}^{(m)}) - f^*(\vec{z}^{(m)}) - \ve R}
\leq \max \{1,
\frac{(\tau - 1) \cdot f^*(\rho \vec{z}^{(m)})}{\frac{1}{\rho} \cdot f^*(\rho \vec{z}^{(m)}) - f^*(\vec{z}^{(m)}) - \ve R} \}.
\]

Observe that the second argument of the maximum operator is exactly the same expression appearing in inequality (\ref{eq:calc2}) apart from the negative $\ve R$ error term in the denominator.
We use the assumption that every term in $f^*$ has degree at least $\lambda = 2$
and set $\rho := 2$.  

We next show how to choose $\ve > 0$ such that $\ve R \leq \frac{1}{10} \cdot f^*(\vec{z}^{(m)})$.  Given the first value function $v_1$,
we can give a lower bound on $\|\vec{z}^{(1)}\|_1$ using the local Lipschitz property of $f^*$.

Let $\beta := \max_{S \subseteq [n]} v_1(S) - v_1(\emptyset)$.
We can assume $\beta > 0$; otherwise, the empty set will be chosen in this round, and $\vec{z}$ and $\vec{x}$ remains zero.

Observe that as $t$ increases from 0 to 1, $\| \rho \vec{z} \|_1 \leq \rho n$.
Suppose $L$ is the local Lipschitz constant of $\nabla f^*$ in the $\ell_1$-ball of radius $\rho n$ around $\vec{0}$. Then, the function $\vec{z} \mapsto \vec{x} := \nabla f^*(\rho \vec{z})$ has local Lipschitz constant $\rho L$.  

Observe that regardless of the value of $\ve$,
if the empty set is ever chosen, the norm $\|\vec{x}\|_1$ must have increased by at least $\beta$, which means $\|\vec{z}\|_1$ has increased by at least $\frac{\beta}{\rho L}$.
On the other hand, if the empty set is never chosen,
then as $t$ increases, $\|\vec{z}\|_1$ increases at rate at least 1.
Hence, in any case, no matter what the value of $\ve$ is, at the end of the first round, $\|\vec{z}^{(1)}\|_1  \geq R_0 := \min \{1, \frac{\beta}{\rho L}\}$.

Hence, after seeing the first value function $v_1$, the algorithm can choose $\ve := \frac{1}{10} \inf_{\|\vec{z}\|_1 \geq R_0} \frac{f^*(\vec{z})}{\|\vec{z}\|_1}$.  Observing that $\vec{z}$ increases monotonically throughout
the algorithm, at the end of round $m$, $\|\vec{z}^{(m)}\|_1 \geq R_0$.
Hence, we have $\ve R \leq \frac{1}{10} \cdot \frac{f^*(\vec{z}^{(m)})}{\|\vec{z}^{(m)}\|_1} \cdot R \leq \frac{1}{10} \cdot f^*(\vec{z}^{(m)})$,
where the last inequality follows because during the time period of measure $R$ in which a non-empty set is chosen, $\|\vec{z}\|_1$ increases at rate at least 1.

Therefore, by a similar analysis as in the proof of Theorem~\ref{thm:packinggeneral},
the competitive ratio is $O(\tau)$, as required.
\end{proof}

{
\bibliography{chk15ref}

\begin{thebibliography}{10}

\bibitem{AlonAABN06}
Noga Alon, Baruch Awerbuch, Yossi Azar, Niv Buchbinder, and Joseph~Seffi Naor.
\newblock A general approach to online network optimization problems.
\newblock {\em ACM Transactions on Algorithms (TALG)}, 2(4):640--660, 2006.

\bibitem{AspnesAFPW97}
James Aspnes, Yossi Azar, Amos Fiat, Serge Plotkin, and Orli Waarts.
\newblock On-line routing of virtual circuits with applications to load
  balancing and machine scheduling.
\newblock {\em Journal of the ACM (JACM)}, 44(3):486--504, 1997.

\bibitem{AwerbuchAGKKV95}
Baruch Awerbuch, Yossi Azar, Edward~F Grove, Ming-Yang Kao, P~Krishnan, and
  Jeffrey~Scott Vitter.
\newblock Load balancing in the l p norm.
\newblock In {\em Foundations of Computer Science, 1995. Proceedings., 36th
  Annual Symposium on}, pages 383--391. IEEE, 1995.

\bibitem{AzarBFP13}
Yossi Azar, Umang Bhaskar, Lisa Fleischer, and Debmalya Panigrahi.
\newblock Online mixed packing and covering.
\newblock In {\em Proceedings of the Twenty-Fourth Annual ACM-SIAM Symposium on
  Discrete Algorithms}, pages 85--100. SIAM, 2013.

\bibitem{AzarCP14}
Yossi Azar, Ilan~Reuven Cohen, and Debmalya Panigrahi.
\newblock Online covering with convex objectives and applications.
\newblock {\em arXiv preprint arXiv:1412.3507}, 2014.

\bibitem{BansalBN12}
Nikhil Bansal, Niv Buchbinder, and Joseph~Seffi Naor.
\newblock A primal-dual randomized algorithm for weighted paging.
\newblock {\em Journal of the ACM (JACM)}, 59(4):19, 2012.

\bibitem{BartalGN03}
Yair Bartal, Rica Gonen, and Noam Nisan.
\newblock Incentive compatible multi unit combinatorial auctions.
\newblock In {\em Proceedings of the 9th conference on Theoretical aspects of
  rationality and knowledge}, pages 72--87. ACM, 2003.

\bibitem{BhawalkarGP14}
Kshipra Bhawalkar, Sreenivas Gollapudi, and Debmalya Panigrahi.
\newblock Online set cover with set requests.
\newblock {\em Approximation, Randomization, and Combinatorial Optimization.
  Algorithms and Techniques (APPROX/RANDOM 2014)}, 28:64--79, 2014.

\bibitem{BlumGMS11}
Avrim Blum, Anupam Gupta, Yishay Mansour, and Ankit Sharma.
\newblock Welfare and profit maximization with production costs.
\newblock In {\em Foundations of Computer Science (FOCS), 2011 IEEE 52nd Annual
  Symposium on}, pages 77--86. IEEE, 2011.

\bibitem{BoydV09}
Stephen Boyd and Lieven Vandenberghe.
\newblock {\em Convex optimization}.
\newblock Cambridge university press, 2009.

\bibitem{BuchbinderCGNN14}
Niv Buchbinder, Shahar Chen, Anupam Gupta, Viswanath Nagarajan, and
  Joseph~(Seffi) Naor.
\newblock Online packing and covering framework with convex objectives.
\newblock {\em arXiv preprint arXiv:1412.8347}, 2014.

\bibitem{BuchbinderG13}
Niv Buchbinder and Rica Gonen.
\newblock Incentive compatible mulit-unit combinatorial auctions: A primal dual
  approach.
\newblock {\em Algorithmica}, pages 1--24, 2013.

\bibitem{BuchbinderN09b}
Niv Buchbinder and Joseph Naor.
\newblock The design of competitive online algorithms via a primal: dual
  approach.
\newblock {\em Foundations and Trends in Theoretical Computer Science},
  3(2--3):93--263, 2009.

\bibitem{BuchbinderN09a}
Niv Buchbinder and Joseph Naor.
\newblock Online primal-dual algorithms for covering and packing.
\newblock {\em Mathematics of Operations Research}, 34(2):270--286, 2009.

\bibitem{DevanurH14}
Nikhil~R Devanur and Zhiyi Huang.
\newblock Primal dual gives almost optimal energy efficient online algorithms.
\newblock In {\em SODA}, pages 1123--1140. SIAM, 2014.

\bibitem{DevanurJ12}
Nikhil~R Devanur and Kamal Jain.
\newblock Online matching with concave returns.
\newblock In {\em Proceedings of the forty-fourth annual ACM symposium on
  Theory of computing}, pages 137--144. ACM, 2012.

\bibitem{DevanurJK13}
Nikhil~R Devanur, Kamal Jain, and Robert~D Kleinberg.
\newblock Randomized primal-dual analysis of ranking for online bipartite
  matching.
\newblock In {\em Proceedings of the Twenty-Fourth Annual ACM-SIAM Symposium on
  Discrete Algorithms}, pages 101--107. SIAM, 2013.

\bibitem{GuptaKP13}
Anupam Gupta, Ravishankar Krishnaswamy, and Kirk Pruhs.
\newblock Online primal-dual for non-linear optimization with applications to
  speed scaling.
\newblock In {\em Approximation and Online Algorithms}, pages 173--186.
  Springer, 2013.

\bibitem{GuptaN14}
Anupam Gupta and Viswanath Nagarajan.
\newblock Approximating sparse covering integer programs online.
\newblock {\em Mathematics of Operations Research}, 2014.

\bibitem{Huang14}
Zhiyi Huang.
\newblock Sigact news online algorithms column 25: Online primal dual: Beyond
  linear programs.
\newblock {\em ACM SIGACT News}, 45(4):105--119, 2014.

\bibitem{HuangK15}
Zhiyi Huang and Anthony Kim.
\newblock Welfare maximization with production costs: A primal dual approach.
\newblock In {\em SODA}, pages 59--72. SIAM, 2015.

\bibitem{YaoDS95}
Frances Yao, Alan Demers, and Scott Shenker.
\newblock A scheduling model for reduced cpu energy.
\newblock In {\em Foundations of Computer Science, 1995. Proceedings., 36th
  Annual Symposium on}, pages 374--382. IEEE, 1995.

\end{thebibliography}
\bibliographystyle{plain}
}

\end{document}